\documentclass[a4paper]{article}
\usepackage{amsmath}
\usepackage{amsfonts}
\usepackage{amssymb}
\usepackage{amsthm}
\usepackage{algorithm,algpseudocode}
\usepackage{wrapfig}
\usepackage{color}
\usepackage{url}
\usepackage{xspace}
\usepackage{makecell}
\usepackage{float}
\usepackage{xcolor}
\usepackage{subcaption}
\usepackage{graphicx}
\usepackage{wrapfig}
\usepackage{xspace}
\usepackage{hyperref}

\newtheorem{theorem}{Theorem}
\newtheorem{definition}{Definition}

\newtheorem{corollary}{Corollary}
\newtheorem{claim}{Claim}
\newtheorem{lemma}{Lemma}

\numberwithin{theorem}{section}
\numberwithin{definition}{section}
\numberwithin{proposition}{section}
\numberwithin{corollary}{section}
\numberwithin{claim}{section}
\numberwithin{lemma}{section}
\numberwithin{obs}{section}

\definecolor{darkred}{rgb}{1, 0.1, 0.3}
\definecolor{darkgreen}{rgb}{0.15, 0.65, 0.15}
\definecolor{darkblue}{rgb}{0.1, 0.1, 1}

\algnewcommand{\algorithmicand}{\textbf{ and }}
\algnewcommand{\algorithmicor}{\textbf{ or }}
\algnewcommand{\OR}{\algorithmicor}
\algnewcommand{\AND}{\algorithmicand}

\newcommand{\vhat} {{\widehat{v}}\xspace}

\newcommand{\Shat} {{\widehat{S}}\xspace}
\newcommand{\chat} {{\widehat{c}}\xspace}

\newcommand{\R} {{\mathbb{R}}\xspace}
\newcommand{\LCA} {{\texttt{LCA}}\xspace}
\newcommand{\Child} {{\texttt{Ch}}\xspace}

\newcommand{\minVertexCover} {\textsc{Minimum Vertex Cover}\xspace}
\newcommand{\maxTwoSAT} {\textsc{Max 2-SAT}\xspace}
\newcommand{\tournamentFeedbackVertexSet} {\textsc{Tournament Feedback Vertex Set}\xspace}
\newcommand{\threeDimensionalMatching} {\textsc{3-dimensional Matching}\xspace}
\newcommand{\minHittingSet} {\textsc{Minimum Hitting Set}\xspace}
\newcommand{\maxAgreementSubtree} {\textsc{Maximum Agreement Subtree}\xspace}
\newcommand{\maxAgreementSupertree} {\textsc{Maximum Agreement Supertree}\xspace}

\newcommand{\SMIS} {\textsc{Strong-MIS}\xspace}
\newcommand{\WMIS} {\textsc{Weak-MIS}\xspace}

\newcommand{\Metric} {{\rho}\xspace}

\newcommand{\matchS}    {{\mathcal{S}}}
\newcommand{\myr}   {{\mathsf{s}}}
\newcommand{\rhat}  {{\widehat{\myr}}}
\newcommand{\xhat}  {{\widehat{x}}}
\newcommand{\yhat}  {{\widehat{y}}}
\newcommand{\zhat}  {{\widehat{z}}}
\newcommand{\myA}   {{\mathsf{A}}}
\newcommand{\myB}   {{\mathsf{B}}}
\newcommand{\myC}   {{\mathsf{C}}}
\newcommand{\myD}   {{\mathsf{D}}}
\newcommand{\myR}   {{\mathsf{S}}}
\newcommand{\myp}   {{\mathsf{p}}}

\title{Ordinally Consensus Subset over Multiple Metrics}
\author{Dingkang Wang, Yusu Wang}

\begin{document}

\setcounter{page}{0}
 
\maketitle

\begin{abstract}
    In this paper, we propose to study the following maximum ordinal consensus problem: Suppose we are given a metric system $(\mathcal{M}, X)$, which contains $k$ metrics $\mathcal{M} = \{\rho_1, \ldots, \rho_k\}$ defined on the same point set $X$. 
    We aim to find a maximum subset $X' \subset X$ such that all metrics in $\mathcal{M}$ are ``consistent'' when restricted on the subset $X'$. In particular, our definition of consistency will rely only on the \emph{ordering} between pairwise distances, and thus we call a ``consistent'' subset an ordinal consensus of $X$ w.r.t. $\mathcal{M}$. 
    We will introduce two concepts of ``consistency'' in the ordinal sense: a strong one and a weak one. Specifically, a subset  $X'$ is strongly consistent means that the ordering of their pairwise distances is the same under each of the input metric $\rho_i \in\mathcal{M}$. The weak consistency, on the other hand, relaxes this exact ordering condition, and intuitively allows us to take the plurality of ordering relation between two pairwise distances. 
    
    We show in this paper that the maximum consensus problems over both the strong and the weak consistency notions are NP-complete, even when there are only 2 or 3 simple metrics, such as line metrics and ultrametrics. We also develop constant-factor approximation algorithms for the dual version, the minimum inconsistent subset problem of a metric system $(\mathcal{M}, P)$, -- note that optimizing these two dual problems are equivalent. 
    \end{abstract}

\newpage 
\section{Introduction}
In recent years, there have been many studies on data sets with multiple views, which can contain different sets of features from multiple sources carrying different types of information. 
For example, consider neuron cells in the field of neuroscience \cite{Multi-view_Neuron_Data2018}. A single neuron cell could have both morphology features and RNA-sequencing information available. Simply concatenating these two types of feature sets and applying a classical single-view method may not produce meaningful results -- 
The types of features may be different, and it is not clear how to properly weigh them when combined. 
Instead, there have been many approaches developed to handle multi-view data. For example, Lashkari and Golland \cite{EM_MV_clustring_Lashkari_and_Golland}, Bickel and Scheffer \cite{EM_MV_clustering_Bickel_and_Scheffer} considered using EM algorithm and (convex) mixture model on multi-view clustering; Kumar, Rai and Hal \cite{MV_spectral_cl_Kumar_and_Rai_and_Hal} and Cai et al. \cite{MV_spectral_cl_Cai_et_al} extended spectral clustering algorithm for multi-view data. 
See also surveys on multi-view clustering \cite{MV_clustering_survey}, and more broadly, on multi-view learning \cite{MV_learning_survey}. 

Very often in applications, multiple views give rise to multiple metrics $\{\Metric_1, \ldots \\, \Metric_k\}$ over the same data set $X$. 
Our goal is to study whether these metrics are ``consistent'' over $X$, and identify a largest subset $X'$ of $X$, called \emph{consensus}, over which these input metrics are ``consistent''. 
However, when comparing these metrics, note that the precise distance values between points in $X$ induced by different $\Metric_i$s may not have the same meaning, two metrics may not have simple, say, linear relation between them, and thus the distance values are not readily comparable (even after normalization). For example, the distance between two neuron cells based on their tree morphology can be very different from that based on their gene expression profiles. 
Hence in this paper, we will compute consensus under multiple metrics based on \emph{ordinal information}, namely the order of pairwise distances under each metric. 

More specifically, given a \emph{metric system} $(\mathcal{M}, X)$, consisting of a set of $k$ input metrics $\mathcal{M} = \{\Metric_1, ..., \Metric_k\}$ on a discrete data set $X$ with cardinality $n$, we propose to study the problem of finding maximum \emph{ordinal consensus of $X$ w.r.t. $\mathcal{M}$}.  
Specifically, we aim to find a maximum subset $X' \subset X$, such that all metrics will have \emph{consistent} pairwise distances if restricted on node set $X'$. We also call $S = X\setminus X'$ as \emph{outliers}, while $X'$ is our targeted ordinal consensus. The dual problem is to find the minimum inconsistent (outlier) set $S$ such that all metrics are consistent when restricted to the subset $X \setminus S$. 

\noindent {\bf Our contributions.} 
We propose two notions to measure the ``ordinal consistency", which we call strong consistency and weak consistency, respectively. Intuitively, a strong consensus $X' \subseteq X$ means that the order of all pairwise distances among $X'$ must be the same w.r.t. all input metrics. 
Under the weak consistency notion, roughly speaking for each pair of pairwise distances, only a plurality of input metrics (instead of all of them) need to agree on that. 
The formal definitions of these consistency notions are in Section \ref{sec:Preliminaries}.

Note that the maximum (ordinal) consensus subset and the minimum inconsistent subset are equivalent. 
In Section \ref{sec:strongly_incon_set} and \ref{sec:weakly_incon_set}, we will show for both the strong and weak consistency definitions, finding the subset $S$ over an input of a constant number of (two for the strong case, and three for the weak case) ultrametrics or Euclidean metrics on the real line are NP-hard. 
These special cases imply that the problems are NP-hard if the inputs are arbitrary metrics. 
We also study the approximation algorithms for both the strongly and weakly minimum inconsistent subset problems. 
In particular, for the strongly inconsistent subset, we propose a 4-approximation algorithm with time complexity $O(kn^2\log n)$. For the weak case, we have a $O(n^6)$-time 6-approximation algorithm. 
See Appendix \ref{app:hardness_results} for a table summarizing our hardness results and approximation algorithms. 

All missing technical details can be found in the appendix. 

\noindent \textbf{Some related work.} 
We note that this maximum consensus problem has been considered before when input metrics are tree-metrics. 
In particular, in the scenario where the inputs are multiple leaf-labeled phylogenetic (rooted) trees, one aims to find a maximum subset of labels that are ``consistent'' among all inputs. 
In \cite{3ordinal_work_liter_max_agreement_subtree_Amir1997}, Amir and Keselman proposed  \maxAgreementSubtree problem (MAST): given a set of rooted binary trees with the same set of taxa (leaf labels), find the maximum subset, such that all the given trees restricted on the subset are isomorphic. This can be considered as a special case of tree consistency \cite{2HCT_Aho_1981}. The \maxAgreementSupertree problem (SMAST) problem is studied in \cite{3ordinal_work_liter_max_agreement_supertree_Berry2007}: Here for the given trees $\mathcal{T} = \{T_1, ..., T_k\}$, the leaf label set $\Lambda(T_i)$ for input trees may not be same. The goal is to find a tree $Q$ with $\Lambda(Q) \subset \cup_{T_i \in \mathcal{T}} \Lambda(T_i)$ such that $|\Lambda(Q)|$ is maximized and for each tree $T_i \in \mathcal{T}$, the subtree $T_i | \Lambda(Q)$ is isomorphic to $Q | \Lambda(T_i)$ (where $T | S$ is the subtree of $T$ restricted on leaf set $S$). 

These definitions of consistency over trees however are not identical from the ordinal consistency we propose. These problems are related to, but still different from, our maximum ordinal consensus problem if the input metrics are ultrametrics. 
An ultrametric can be represented by a corresponding representing tree where each tree node has a height value. Finding a strong consensus $S$ is equivalent to finding a subset of leaf nodes such that the restricted subtrees are not only isomorphic, \emph{but also} the heights of all internal nodes must have the same order -- This height condition appears to make the problem much harder: While the MAST problem on two trees can be solved in polynomial time via dynamic programming, the ordinal consensus problem is NP-hard even for only two ultrametrics, as we will show in Section \ref{sec:strongly_incon_set}. 


\section{Preliminaries and problem setup} \label{sec:Preliminaries}
    The input is a \emph{metric system $(\mathcal{M}; X)$}, consisting of a set of $k$ metrics $\mathcal{M} = \{\Metric_1, \Metric_2, ..., \Metric_k\}$ over the point set $X = \{x_1, ..., x_n\}$. For any $i \in [1,k]$, $\Metric_i(x, x')$ is the distance between point $x, x'\in X$ w.r.t. $\Metric_i$. 
    Our goal is to find a minimum subset $S \subset X$, such that the order of all pairwise distance restricted on $X \backslash S$ are consistent under a certain definition. Below we first introduce two notions of consistency. The two optimization problems we will study are given in Definition \ref{def:MSI_MWI}. 
    
    
    \begin{definition}[Strong Consistency]
        Given a metric system $(\mathcal{M} = \{\Metric_1, \Metric_2, ..., \\ \Metric_k\}; Y = \{y_1, y_2, ..., y_m\})$, we say that the set of metrics $\mathcal{M}$ is {strongly consistent w.r.t. $Y$} if for any quartet $\{y_p, y_q, y_r, y_s\} \subset Y$, we have that (i) $ \Metric_{i}(y_p, y_q) < \Metric_{i}(y_r, y_s) \Leftrightarrow \Metric_{j}(y_p, y_q) < \Metric_{j}(y_r, y_s)$ and (ii) $\Metric_{i}(y_p, y_q) = \Metric_{i}(y_r, y_s) \Leftrightarrow \Metric_{j}(y_p, y_q) = \Metric_{j}(y_r, y_s)$ for any $1 \le i, j \le k$. In this case, we say that $Y$ is \emph{a strongly consistent set, or a strong consensus, over $\mathcal{M}$}. 
        
        We also say that two pairs $(y_p, y_q)$ and $(y_r, y_s)$ are \emph{strongly consistent w.r.t. $\mathcal{M}$}, if the order between these two pairwise distances is the same w.r.t. any metric in $\mathcal{M}$.
    \end{definition}
    
    Other than strong consistency, we also consider a weaker notion of consistency: In particular, we now only require that the order constructed by taking the \emph{plurality voting}\footnote{In plurality voting, a candidate wins if it has the most votes than the other candidates. It does not have to get a majority (more than 50\%) of the votes. } over all input metrics is valid. 
    To define the weak consistency formally, we will first define the so-called relation set and the auxiliary graph. 

    \begin{definition}[Relation Set of Pairwise Distances] \label{def:3relation_set_of_pd}
        Given an input set of $k$ metrics $\mathcal{M}$ over point set $X$, the \emph{relation set $\mathcal{R}$ of pairwise distances} w.r.t. $(\mathcal{M}; X)$ is the set of relations over all distinct pairs $\{(x_p, x_q) \, | \, x_p, x_q \in X, p \not= q\}$ defined as follows: 
        For any two pairs $(x_p, x_q)$ and $(x_r, x_s)$, among all three possible relations between $(x_p, x_q)$ and $(x_r, x_s)$, namely, $``<"$, $``="$, and $``>"$, the one induced by most number of metrics in $M$ will be included in the relation set $R$. 
        If there is any tie, we will choose the relation with appearance in the metric of smaller index. 
    \end{definition}
    
    For example, the relation $(x_p, x_q) < (x_r, x_s) \in \mathcal{R}$ \emph{if and only if} this relation appears in more (or the same number of) metrics than the other two relations ($``="$, $``<"$).  
            
    The relation set constructed above may be not ``valid'' in the sense that no single metric can generate those relations. To check the validity of this relation set, we now define a specific auxiliary graph $\mathcal{G}$, whose nodes correspond to pairs of points from $X$. There are three different connections between two graph nodes, which correspond to the three possible relations between their corresponding pairs in the relation set $\mathcal{R}$. We will use this graph later to decide whether $\mathcal{R}$ is valid or not.
    
    \begin{definition}[Auxiliary Graph for Relation Set]\label{def:auxiliarygraph} 
        Given the relation set $\mathcal{R}$ of pairwise distances over point set $X$ and metrics $\mathcal{M}$, the auxiliary graph $\mathcal{G} = (V, E)$, where $V = \{(x_i, x_j)| x_i, x_j \in X, i \not= j\}$, is a mixed graph (meaning it contains both directed and undirected edges):
        There is a \emph{directed} edge from $v_1 = (x_p, x_q)$ to $v_2 = (x_r, x_s)$ if $(x_p, x_q) > (x_r, x_s) \in \mathcal{R}$; there is a \emph{undirected} edge between $v_1$ and $v_2$ if $(x_p, x_q) = (x_r, x_s) \in \mathcal{R}$. 
    \end{definition}
    
    We will use $v_1, ..., v_N$ (corresponding to all pairs in $X$) to represent nodes in graph $\mathcal{G}$, where $N = {n \choose 2}$ if $n = |X|$. The auxiliary graph $G$ is a fully connected mixed graph (i.e., every pair of distinct vertices is connected by a unique edge, and the edge can be directed or undirected) with one edge between every two nodes. 
    We say $v_i > v_j$ (or $``="$ or $``<"$) if there is a directed edge from $v_i$ to $v_j$ (or an undirected edge between them, or a directed edge from $v_j$ to $v_i$), which means the pair, say $(x,y)\in X \times X$, represented by $v_i$ has a larger pairwise distance compared with the pair represented by $v_j$. 
    A cycle in a mixed graph can be formed by a mixture of directed edges and undirected edges. 
    Intuitively, suppose we have a completely directed cycle like the one shown in Figure \ref{fig:directed_edges_mixed_graph} (a), 
    then there is a conflict in the relations for pairwise distances, as $v_1 < v_2 < v_3 < v_4 < v_1$. Hence the relation set $\mathcal{R}$ will {\bf not} be valid in this case. 
    There are more types of ``directed" cycles and they can cause such conflict. In particular: 

    \begin{definition}[Directed cycle of mixed graphs]
        A cycle $C = \{e_1, \ldots, e_r\}$ is a \emph{fully-directed cycle} if all edges in it are directed, and the directions of all edges inside are consistent (namely, Figure \ref{fig:directed_edges_mixed_graph} (a)). 
        A directed cycle of a mixed graph is a cycle $C = \{e_1, e_2, ..., e_n\}$ such that (i) it consists of at least one directed edge, and (ii) one can assign a direction for each undirected edges in $C$ to make it into a fully directed cycle. See Figure \ref{fig:directed_edges_mixed_graph} for examples.  
    \end{definition}
    
    \begin{figure*}[h]
    \begin{center}
     \resizebox{\textwidth}{!}{
        \begin{tabular}{cccc}
         \includegraphics[height=0.7in]{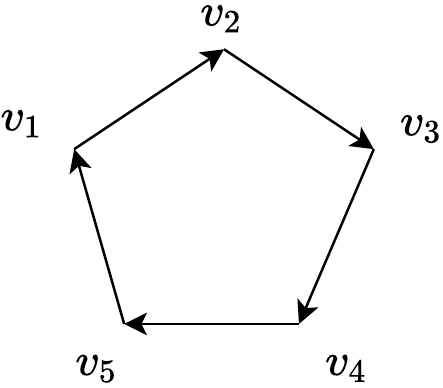} & 
         \includegraphics[height=0.7in]{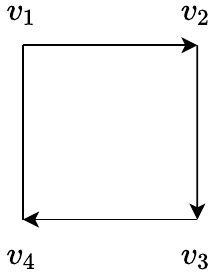} &
         \includegraphics[height=0.7in]{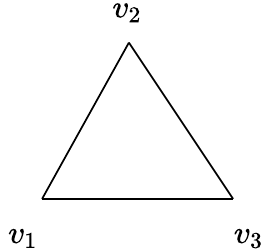} &
         \includegraphics[height=0.7in]{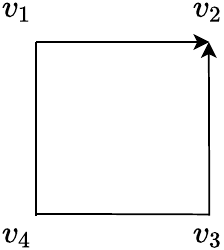} \\
         
        (a) A full-directed  & (b) A directed cycle & (c) This is not  & (d) This is not  \\
         (and directed) cycle &  (of mixed graphs) with & a directed cycle since  &  a directed cycle since 
         \\
         with length 5. &  an undirected edge. & none of the & conflicts between
         \\
         &&edges is directed. & directed edges.
         
        \end{tabular}
        }
    \end{center}
    \caption{Examples of fully-directed cycle, directed cycle and other cycles of a mixed graph.}
    \label{fig:directed_edges_mixed_graph}
\end{figure*}
    
    We say that a metric $\widehat{\Metric}$ \emph{generates} a relation set $\mathcal{R}$, if the order of any two pairwise distances induced by $\widehat{\Metric}$ are the same as in $\mathcal{R}$. In this case, we say the relation set $\mathcal{R}$ is \emph{valid}.  
    The following result provides a simple characterization for a valid relation set by using its corresponding auxiliary graph. 
    \begin{lemma}\label{lem:validrelationset}
            There exists a metric $\widehat{\Metric}$ generating a  relation set $\mathcal{R}$, or equivalently, $\mathcal{R}$ is valid, if and only if there is no directed cycle in the auxiliary graph $\mathcal{G}$ for the relation set $\mathcal{R}$. 
    \end{lemma}
    \begin{proof}
    
    \noindent \textbf{The ``$\Rightarrow$" direction.} This direction is relatively easier and we prove it by contradiction. Suppose the relation set is defined over a point set $\widehat{X}$ and there is a metric $\widehat{\rho}$ over set $\widehat{X}$ that generates all relations in $\mathcal{R}$. Then assume there is a directed cycle of graph nodes $C = \langle v_1, v_2 \cdots v_r, v_1\rangle$ in the auxiliary graph $\mathcal{G}$. Now, one can derive that the distance between the pair of nodes in $v_1$ will be larger than itself by following the edges $v_i\to v_{i+1}$ for $i = 1, \ldots, r$, which is a contradiction.  

   	\noindent \textbf{The ``$\Leftarrow$" direction.}
    For any two nodes $v_a$ and $v_b$ connected by an undirected edge (i.e., $v_a$, $v_b$ represents equal distance), we can always contract them together without causing any problems if there is no directed cycle in graph $\mathcal{G}$. It is because that for another arbitrary $v_c$, the directions of edges $(v_a, v_c)$ and $(v_b, v_c)$ are consistent. These two edges will be both going to $v_c$, coming out from $v_c$ or undirected.
   	
    Therefore, we can iteratively contract all node pairs connected by undirected edges until there is no such pair left. The produced graph $\widehat{\mathcal{G}}$ will be a directed graph, and each node corresponds to a set of nodes in the original graph $\mathcal{G}$. It is clear that $\widehat{\mathcal{G}}$ is acyclic, because any directed cycle in $\widehat{\mathcal{G}}$ will map back to a directed cycle (of a mixed graph) in $\mathcal{G}$ simply by selecting an arbitrary representative for all these supernodes in $\widehat{\mathcal{G}}$.
   	
    A linear order over those supernodes can be constructed by, e.g., topological sort, such that the orders among pairwise distances are consistent with the directed acyclic graph (DAG) $\widehat{\mathcal{G}}$. Nodes in $\mathcal{G}$ represented by the same supernode in $\widehat{\mathcal{G}}$ share the same value. This provides a way to assign values over all pairwise distances such that the ordinal relations are consistent with $\mathcal{G}$. To make it a metric, one can force that the minimum distance is larger than one half of the maximum distance.
    This way, no matter which three nodes we consider, the sum of two pairwise distances will always be larger than the third pairwise distance. 
    \end{proof}
    
    We are now ready to define the weak consistency. 
    
    \begin{definition}[Weak Consistency] \label{def:weakly_consistent_met}
        Given a set of $k$ metrics $\mathcal{M} = \{\Metric_1, \Metric_2, ..., \Metric_k\}$ over the same point set $Y = \{y_1, y_2, ..., y_m\}$, we say that the set of metrics $\mathcal{M}$ is weakly consistent with $Y$ if there is no directed cycle in the auxiliary graph $\mathcal{G}(Y)$ for the relation set $\mathcal{R}(Y)$ as specified in Definitions 
        \ref{def:3relation_set_of_pd} and \ref{def:auxiliarygraph}. In this case, we may also say that $Y$ is a weakly consistent set, or a weak consensus, over $\mathcal{M}$.
    \end{definition}

The optimization problems we aim to study in this paper are defined as follows. 
    
    \begin{definition}[\SMIS and \WMIS Problems]\label{def:MSI_MWI}
        Given a set of metrics $\mathcal{M} = \{\Metric_1, ..., \Metric_k\}$ on point set $X$, the \SMIS problem (resp. \WMIS problem) aims to find a minimum subset of $S^* \subset X$, such that all metrics restricted on $X \backslash S^*$ are strongly consistent (resp. weakly consistent).
        
        The set $S^*$ is called the \emph{minimum (strong/weak) inconsistent set}, while $X \setminus S^*$ is called the \emph{maximum (strong/weak) consensus set w.r.t. input metrics $\mathcal{M}$}. 
    \end{definition}
 
 Note that minimizing the inconsistent set is equivalent to maximizing the consensus (although their approximation may not be equivalent). 
 
The decision version of $\SMIS$ (and $\WMIS$) problem is as follows: 
Given $\mathcal{M}, X$ and also an integer $a$. Is there a subset $S \subset X$ with $|S| = a$ such that $\mathcal{M}$ restricted on $X \backslash S$ are strongly consistent (resp. weakly consistent)? 

We will show in Section \ref{sec:strongly_incon_set}
and \ref{sec:weakly_incon_set} that the decision versions of $\SMIS$ and $\WMIS$ are in NP (see lemma \ref{lem:smis_in_NP} and lemma \ref{lem:wmis_in_NP}). Thus in most proofs, we will only show NP-hardness via reductions from NP-hard problems. The NP-completeness naturally follows by the fact that both decision problems are in NP.

\paragraph*{Some specific metrics.}
Later in Sections
\ref{sec:strongly_incon_set} and \ref{sec:weakly_incon_set}, we will show that the decision version of the minimum inconsistent set (equivalently, maximum consensus set) problem is NP-complete even when input metrics are restricted to two common choices: the Euclidean metric on the line (and thus in any $\R^d$), and the ultrametrics. 

\begin{definition}[Line metric] \label{def:line_metric}
    A line metric is a metric $(\R, \Metric)$, where the distance function $\Metric(x, y) = |x - y|$, and $x, y \in \R$ (Note, this is simply the Euclidean metric on $\R$).
\end{definition}

\begin{definition}[Ultrametric] \label{def:ultrametric}
    An ultrametric is a metric $(Z, \Metric)$ defined on a set $Z$, 
    which satisfies the following strong triangle inequality: for any $x, y, z \in Z$, 
    $d(x, z) \le \max (d(x, y), d(y, z)).$
\end{definition}

Any finite ultrametric $(Z, \Metric)$ has a corresponding \emph{representing tree} \cite{3Ordinal_representing_tree_Dovgoshey2019} $T_Z$ such that: 
    \begin{enumerate}
        \item $T_Z$ is a rooted tree with the set of leaf nodes being $Z$. $T_Z$ is equipped with a height function $h: N \cup Z \rightarrow \mathbb{R}_{+}$, where $N$ is the set of internal nodes of $T_Z$ such that (i) all leaves have the same height; and (ii) $h$ is non-increasing along any root to leaf path.
        \item For any two leaf nodes $z$ and $z'$, their distance $\Metric(z, z')$ equals to $h(\LCA(z, z'))$, namely, the height of their lowest common ancestor (\LCA($z,z'$)).  
    \end{enumerate}
    
    An example of an ultrametric and its representing tree is given in Appendix \ref{app:ultra_representingTree}.
    
\section{\SMIS problem}
    \label{sec:strongly_incon_set}

    In this section, we will study the \SMIS problem. 
    Specifically, we show in Theorem \ref{thm:MSI_2_line_mets} and \ref{thm:MSI_2_ultra_NP} that the decision version of the \SMIS problem is NP-complete even when the input metric spaces are restricted to two very simple cases: the line metrics and the ultrametrics. 
    Corollary \ref{coro:strong_NP_2metrics_inapprox} gives an inapproximability result. 
    To complement that, in Theorem \ref{thm:MSI_4_approx_alg}, we provide a $4$-approximation algorithm for the general case.

    Intuitively, finding the minimum strongly inconsistent set has a similar flavor as \minVertexCover or \minHittingSet \cite{HittingSet_NP_Book_Garey}. Intuitively, in \SMIS problem, given data points $X$, a quartet $(x_p, x_q, x_r, x_s)$ will be a target set if they induce a conflict between any two input metrics, and the goal is to find a \minHittingSet $\mathcal{H}$ such that for all target sets, at least one element is in $\mathcal{H}$. 
    However, to show that \SMIS remains hard even for special simple metrics, we need to construct reductions carefully, and sometimes need to use different NP-complete problems to reduce from. We include the list of NP-complete problems used for reductions in Appendix \ref{app:np_complete_repo}. 

    To study \SMIS, we first define the so-called Conflict Set, which will be used frequently in NP-hardness proofs.
    
    \begin{definition}[Conflict Set] \label{def:conflict_set}
        Given $k$ metrics $\mathcal{M} = \{\Metric_1, ..., \Metric_k\}$ on node set $X = \{x_1, ..., x_n\}$, the \emph{conflict set $\mathcal{C}$} induced by $(\mathcal{M}; X)$ is defined as $\mathcal{C} = \{(x_p, x_q, x_r, x_s) | (x_p$, $x_q), (x_r, x_s) \text{ are not strongly consistent over } \mathcal{M}\}$. 
        Each element in this conflict set is called a \emph{conflict quartet}. 
    \end{definition}
    
    It is clear that the decision version of $\SMIS$ is in NP as stated in the following lemma.
    
    \begin{lemma} \label{lem:smis_in_NP}
    The decision version of \SMIS is in NP.
    \end{lemma}
    \begin{proof}
    Given a set $S$ with size $a$, one can check whether metrics in $\mathcal{M}$ are strongly consistent on $X \backslash S$ by simply iterating over all possible quartets and comparing their pairwise distances in different metrics. This process takes polynomial time.
    \end{proof}
    
    
    It turns out that the decision version of \SMIS is weakly NP-complete (which allows the magnitude of data involved to be exponential) even if we restrict the input metrics to be only two line metrics. The proof is in Appendix \ref{app:MSI_2_line_mets_NP}.
    
    \begin{theorem} \label{thm:MSI_2_line_mets}
        The decision version of \SMIS is weakly NP-complete even when one only considers metric systems $(\mathcal{M}; X)$ where $\mathcal{M} = \{\Metric_1, \Metric_2\}$ contains only two line metrics. 
    \end{theorem}

    The following theorem shows that even with only 2 ultrametrics, finding a \SMIS remains NP-hard. The proof uses a reduction from \maxTwoSAT problem; it is non-trivial and can be found in Appendix \ref{app:MSI_2_ultra_NP}. 
    
    \begin{theorem} \label{thm:MSI_2_ultra_NP}
            Given a metric system $(\mathcal{M}; X)$, where $\mathcal{M} = \{\Metric_1, \Metric_2\}$ contains two ultrametrics, The decision version of \SMIS is NP-complete.
    \end{theorem}

        The hardness result on 2 ultrametrics implies that finding \SMIS is also NP-hard when the input is two arbitrary metrics. 
        The result is stated below. This theorem can also be proven directly via a reduction from \minVertexCover, which for completeness we include the simple details in Appendix \ref{app:prove_thm_strongNP_2metric}. 
        
        \begin{theorem} \label{thm:strong_NP_2metrics}
            Given a metric system $(\mathcal{M}; X)$, where $\mathcal{M} = \{\Metric_1, \Metric_2\}$ contains two arbitrary metrics, the decision version of \SMIS is NP-complete.  
        \end{theorem}
        
        In fact, the proof in Appendix \ref{app:prove_thm_strongNP_2metric} gives a size-preserving reduction. 
    Hence the Corollary below follows directly from the inapproximability result \cite{3ordinal_work_VC_inapprox_Subhash08} for \minVertexCover.     
        \begin{corollary} \label{coro:strong_NP_2metrics_inapprox}
            \SMIS with 2 metrics is Unique Games-hard to approximate within a factor $2 - \epsilon$, where $\epsilon$ is an arbitrarily small positive number. 
        \end{corollary}

\paragraph*{Approximation algorithm for \SMIS.}
        As we mention above, one can consider a collection $\mathcal{C}$ of all conflict quartet $(x_p, x_q, x_r, x_s)$ as the target set. The goal is to find a minimum set from $X$ such that it intersects (hits) every quartet in $\mathcal{C}$. This is actually a special case of $4$-hitting set problem, and it is easy to obtain a 4-approximation algorithm in time $O(k n^4)$ time by checking all quartets. However, below we show we can improve this to $O(kn^2 \log n)$ time complexity.
    
    \begin{theorem} \label{thm:MSI_4_approx_alg}
        Given a metric system $(\mathcal{M}; X)$ where $\mathcal{M} = \{\Metric_1, \ldots, \Metric_k\}, X = \{x_1,\ldots, x_n\}$, there is an $O(kn^2\log n)$ 4-approximation algorithm for the \SMIS problem.
    \end{theorem}
    \begin{proof}
    Let $S^*$ denote the minimum inconsistent set so that $X \setminus S^*$ is the maximum consensus for metric system $(\mathcal{M}; X)$. 
    We propose Algorithm \ref{alg:strong_MSI}, which will compute a set $S$ to be removed as inconsistent set. 
         \begin{algorithm}[h] 
	        \caption{$S$ = \SMIS($\mathcal{M} = (\Metric_1, \cdots, \Metric_k), X$)}
		    \label{alg:strong_MSI}
		    \begin{algorithmic}[1]
		    \For{each metric $\Metric_i \in \mathcal{M}$} 
    		    \State Sort all pairwise distances based on the distances in ascending order. Let $L_1, \cdots, L_k$ denote those $k$ sorted lists.
		    \EndFor
    		\State Initialize $S = \emptyset.$ $k$ pointers $\myp_1, \cdots, \myp_k$ pointing to the head of $L_1, \cdots, L_k$.
	    	\While{None of $\myp_i$s points out of bound.}
	    	  \While{$E(\myp_1) \cap S \not= \emptyset$} 
	    	  \State $\myp_1 = \myp_1 + 1$  
		      \EndWhile \Comment{Move the pointer of the first list.}
		      \State flag = {\texttt{False}}
		      \For{$i \in [2 \cdots k]$}
		      \While{$E(\myp_i) \cap S \not= \emptyset$}   
		       \State $\myp_i = \myp_i + 1$  
		       \EndWhile
		       \If{$\myp_1, \myp_i$ in bound $\AND E(\myp_i) \not= E(\myp_{1})$}
		            \State $S = S \cup \{E(\myp_i), E(\myp_{1})\}$
		        \State flag = {\texttt{True}}             \Comment{Found a conflict}
		       \EndIf
		    \State $\myp_i = \myp_i + 1$ \Comment{Move the pointer to the next since it is already considered.}
		    \If{flag = {\texttt{True}}}   \Comment{If there is a conflict, move on to the next iteration.}
		    \State \textbf{break}
		    \EndIf
		    \EndFor
		\EndWhile
    		\State \textbf{return $S$}
    	\end{algorithmic}
    \end{algorithm}
        
        On the high level, Algorithm \ref{alg:strong_MSI} starts with sorting all pairs of points based on their distances (If two pairs have the same distance, then they are sorted in lexicographical order). 
        With $k$ sorted lists of pairs $L_1,\ldots, L_k$, we set $k$ pointers to the heads of these lists: Let $\myp_i$ be the pointer to the list of pairwise distances of metric $\rho_i$, and $E(\myp_i)$ denote the pair of data points in $X$ that $\myp_i$ is pointing to. 
        As pointers move down these lists, it checks whether all pointers are pointing to the same pair of points from $X$. If that is the case, it moves on. Otherwise, if $E(\myp_i) \neq E(\myp_{1})$ (line-14), it means a conflict quartet is discovered, formed by the two pairs $E(\myp_i)$ and $E(\myp_{1})$. In this case, we will remove all these points (by adding them to $S$), and move on. 
        The code from lines 6-8 and 11-13 is to skip all pairs which contain at least one point from the outlier set $S$. 
        The procedure ends when any pointer moves out of bound (i.e., beyond the last element of the list).
        An illustration of how the algorithm works is given in Figure \ref{fig:SMIS_4_approx_alg} in Appendix \ref{app:missing_figure_MSI_4approx_alg}.

        \noindent \textbf{Time complexity.}
        There are $k$ metrics and $n \choose 2$ pairs. The sorting process takes $O(k n^2 \log n)$ time.
        Since the procedure ends when each pointer reaches the tail, the pointers will be moved for at most $k n^2$ total times. During each pointer move, the most expensive step is line-11, which can be done in $O(1)$ with an array of length $n$ indicating if a node is in $S$ or not. 
        
        Putting everything together, the total time complexity for our algorithm is $O(kn^2 \log n)$.
        
        \noindent \textbf{Correctness of algorithm.}
        We consider the points ever added to set $S$. 
        Note that $S$ is only updated in line-12, where there is a conflict (i.e, the pair $E(\myp_i) = (a, b)$ pointed by $p_i$ in list $L_i$ is different from the pair $E(\myp_{1}) = (c,d)$ by the pointer for list $L_{1}$). Assume that $(a, b)$ is smaller than $(c, d)$ for lexicographical order. 
        We claim that $\{(a,b), (c,d)\}$ form a conflict quartet. 
        To prove this, first, note that as we skip all pairs that contain any element from $S$ (lines 6-8, 11-13), this means that $S \cap \{a, b, c, d\} = \emptyset$. 
        Hence in list $L_{1}$, we have not yet seen (scanned) the pair $(a,b)$ -- as otherwise, at the time when the pointer in $L_{1}$ reaches $(a,b)$, if at that moment the pointer in $L_i$ is not pointing to $(a,b)$, we would have already seen a conflict and added $a, b$ to $S$. 
        It then follows that w.r.t. metric $\Metric_{1}$, we have that $\Metric_{1}(a,b) > \Metric_{1}(c,d)$. 
        On the other hand, in list $L_i$, it must be that we have not yet seen $(c,d)$ by the same reasoning, meaning that $\Metric_i(a,b) \le \Metric_i(c,d)$ w.r.t. metric $\Metric_i$. Hence these two pairs form a conflict quartet. 
        Obviously, for any conflict quartet, the minimum inconsistent set $S^*$ has to contain at least one element from it. 
        Furthermore, since all the conflict quartet the algorithm ever identifies are disjoint. This means that the $|S^*| \ge |S|/4$, that is, $|S| \le 4|S^*|$. 
        
        Finally, consider the ordered sublist $\widehat{L}_i$ of $L_i$, obtained by removing from $L_i$ all pairs that intersect $S$. 
        Then it is easy to see that by construction of the algorithm, all $\widehat{L}_i$s are the same. In other words, after removing all elements in $S$, the remaining points $X \setminus S$ form a consensus subset for the $k$ metrics $\mathcal{M} = \{\Metric_1, \ldots, \Metric_k\}$. Hence $S$ is a valid inconsistent set and $|S| \ge |S^*|$. 
        It follows that $S$ is a 4-approximation of the minimum inconsistent subset for the metric system $(\mathcal{M}; X)$. 
        
     \end{proof}

    \section{\WMIS problem} \label{sec:weakly_incon_set}
    
    We now focus on the \WMIS problem of finding a minimum weakly inconsistent subset.  
    In Theorems \ref{thm:MWI_3_line_metrics} and \ref{thm:MWI_3_ultra_NP}, we show that it is NP-complete for the special case of only three input line metrics or ultrametrics. We provide a straightforward 6-approximation algorithm at the end.  
    
    \begin{figure}[thbp]
        \centering 
        \includegraphics[width=0.8\textwidth]{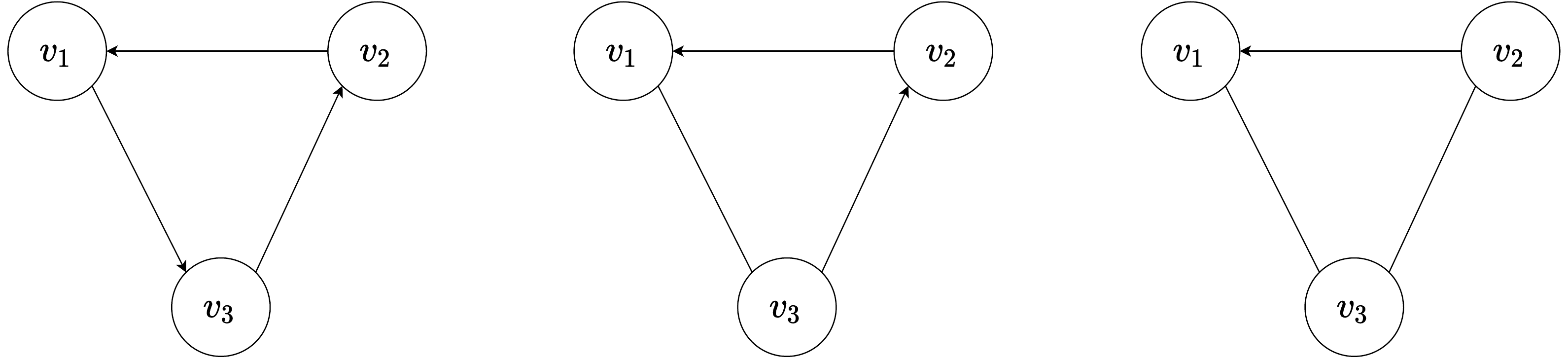}
        \caption{Three possible directed triangles.}
        \label{fig:directed_triangles}
    \end{figure}
    By definition \ref{def:weakly_consistent_met} and \ref{def:MSI_MWI}, if $X \backslash S$ is a consensus set, then the auxiliary graph $\mathcal{G}$ restricted on $X \backslash S$  must contain no directed cycle. 
    It is well known that a tournament (fully connected directed graph) has a directed cycle if and only if it has a directed triangle \cite{3ordinal_directed_tri_Dom2010}. It turns out that a similar result holds for auxiliary graphs, which are mixed graphs: see the claim below. 
See Figure \ref{fig:directed_triangles} for the 3 possible cases of directed triangles.
The simple proof of this claim can be found in Appendix \ref{app:prove_claim_3_directed_tri}.
    \begin{claim}\label{claim:3_directed_tri}
        An auxiliary graph $\mathcal{G}$ has no directed cycle if and only if it has no directed triangle.
    \end{claim}

    Hence to see whether there is any directed cycle in the auxiliary graph, one only needs to check if there is any directed triangle. 
    
Followed by claim \ref{claim:3_directed_tri}, the decision version of \WMIS is in NP.
\begin{lemma} \label{lem:wmis_in_NP}
The decision version of \WMIS is in NP.
\end{lemma}
\begin{proof}
    By definition \ref{def:weakly_consistent_met} and claim \ref{claim:3_directed_tri}, one can check whether metrics from $\mathcal{M}$ are weakly consistent on a set $S$ by iterating over all triangles in the auxiliary graph. It is clearly polynomial.
\end{proof}

When there are two metrics, with the tie-breaking rule defined in Definition \ref{def:3relation_set_of_pd}, it is clear that we would always prefer the first metric (once the order of input metrics is fixed). Thus in this case, the minimum inconsistent set is simply $\emptyset$. The problem of \WMIS becomes non-trivial when there are three metrics.
 
 Our first main result is as follows, with proof in Appendix \ref{app:prove_MWI_3_line_mets}.
    \begin{theorem} \label{thm:MWI_3_line_metrics}
        Given a metric system $(\mathcal{M}; X)$, where $\mathcal{M} = \{\Metric_1, \Metric_2, \Metric_3\}$ contains three line metrics. The decision version of \WMIS is weakly NP-complete.
    \end{theorem}
    
    Our second main result is the hardness for ultrametrics. 
    \begin{theorem}\label{thm:MWI_3_ultra_NP} 
        Given a metric system $(\mathcal{M}; X)$, where $\mathcal{M} = \{\Metric_1, \Metric_2, \Metric_3\}$ contains three ultrametrics. The decision version of \WMIS is NP-complete.
    \end{theorem} 
    \begin{proof}[Proof of Theorem \ref{thm:MWI_3_ultra_NP}]
        We prove this theorem via a reduction from the so-called \threeDimensionalMatching problem. 
    In particular, instead of the problem of finding the minimum inconsistent set, we will consider the equivalent dual version of finding a maximum consensus for a set of 3 ultrametrics. 
    
    \noindent \textbf{Description of the reduction.}
    Suppose we are given an instance of \threeDimensionalMatching problem ($X, Y, Z; \matchS \subseteq X \times Y \times Z$), where $|X| = |Y| = |Z|$. Assume that $X = \{x_1, ..., x_n\}, Y = \{y_1, ..., y_n\}, Z = \{z_1, ..., z_n\}$, while $\matchS = \{\myr_1, ..., \myr_m\}$ where each relation $\myr_i$ is of the form $\myr_i = (x_a, y_b, z_c)$ with $x_a \in X, y_b \in Y$ and $z_c\in Z$. A matching $\Pi \subset \matchS$ is such that each element in $X \cup Y \cup Z$ can appear at most once in all relations in $\Pi$. 
    The decision version of the \threeDimensionalMatching problem is that, given $(X, Y, Z, \matchS)$ and an integer $K$, does there exist a matching $\Pi \subset \matchS$ such that $|\Pi| = K$?  
    
    From this instance $(X, Y, Z; \matchS)$ of the \threeDimensionalMatching problem, we will now construct an instance of the \WMIS problem $(\mathcal{M} = \{U_X, U_Y, U_Z\}, P)$, where the node set is $P = \{a_1, \ldots, a_m, b_1, \ldots, b_m, c_1, \ldots, c_m, \\ d_1, \ldots, d_m, \rhat_1, \ldots, \rhat_m\}$, and $\mathcal{M}$ consists of 3 ultrametrics, $U_X, U_Y$, and $U_Z$ over node set $P$. 
    (Note that we do not use $X$ as the node set as $X$ is already used in the instance of \threeDimensionalMatching problem.) 
     Recall that any ultrametric over node set $P$ corresponds to a representing tree, which is a rooted tree where all nodes have a height value, and all leaves (corresponding to node set $P$) have the same height. In what follows, we will describe the three representing trees $T_X, T_Y$ and $T_Z$, generating $U_X, U_Y$ and $U_Z$, respectively. 
    In particular, these three representing trees $T_X, T_Y$ and $T_Z$ all have the same tree shape. However, the height of those internal nodes will be different. 
    
    We will first describe the representing tree for $T_X$.  
    The root $root$ has 5 children, represented by $A, B, C, D$ and $\widehat{X}$. 
    Each of $A, B, C, D$ has exactly $m$ children, which are $\Child(A) = \{a_1, \ldots, a_m\}$, $\Child(B) = \{b_1, \ldots, b_m\}$, $\Child(C) = \{c_1, 
    \ldots, c_m\}$ and $\Child(D) = \{d_1, \ldots, d_m\}$, respectively. 
    Note that these children are leaves, corresponding to the first $4m$ nodes in the node set $P$ (See Figure \ref{fig:3DM_MWI_exp}). 
    The node $\widehat{X}$ has $n$ children, $\xhat_1, \ldots, \xhat_n$, corresponding to the $n$ points in input set $X$ of the \threeDimensionalMatching instance $(X, Y, Z; \matchS)$. 
    The child(ren) of each $\xhat_i$ is defined as: 
    $\Child(\xhat_i) = \{ \rhat_j \mid x_i \in \myr_j \}$; all children of $\xhat_i$s are all leaves. 
    See Figure \ref{fig:3DM_MWI_exp}. 
    
    Next, we assign height values for nodes in $T_X$. 
    All leaves (corresponding to elements in the node set $P$ where ultrametrics are defined on) have height $0$. The height values for the internal nodes are listed in the row corresponding to $T_X$ in Table \ref{tab:value_assigned}. 
    
    The representing tree $T_Y$ (resp. $T_Z$) has the same tree shape as $T_X$, and the only difference is that the node $\widehat{X}$ and $\xhat_i$s are replaced by $\widehat{Y}$ and $\yhat_i$s (resp. by $\widehat{Z}$ and $\zhat_i$s). See Figure \ref{fig:3DM_MWI_exp}. 
    The height values of all leaves nodes are still $0$, and the height values of internal nodes are listed in the last two rows of Table \ref{tab:value_assigned}. Also see Figure \ref{fig:3DM_MWI_exp} where the height of each node is listed in the parenthesis next to each node. 
    
        \begin{table}[ht] 
            \centering
            \begin{tabular}{llllllll}
                & $root$  & $A$ & $B$ & $C$ & $D$ & $\widehat{X}(\widehat{Y}, \widehat{Z})$ & $\widehat{x}_i (\widehat{y}_i, \widehat{z}_i)$ \\
                $T_X$ & 10 & 5 & 4 & 3 & 1 & 2       & 0                \\
                $T_Y$ & 10 & 3 & 0 & 2 & 5 & 4       & 1                \\
                $T_Z$ & 10 & 2 & 4 & 0 & 1 & 5       & 3               
            \end{tabular}
            \caption{Height function values assigned to internal nodes.} \label{tab:value_assigned}
        \end{table}
        
        \begin{figure}[H]
            \centering 
            \includegraphics[width=0.85\textwidth]{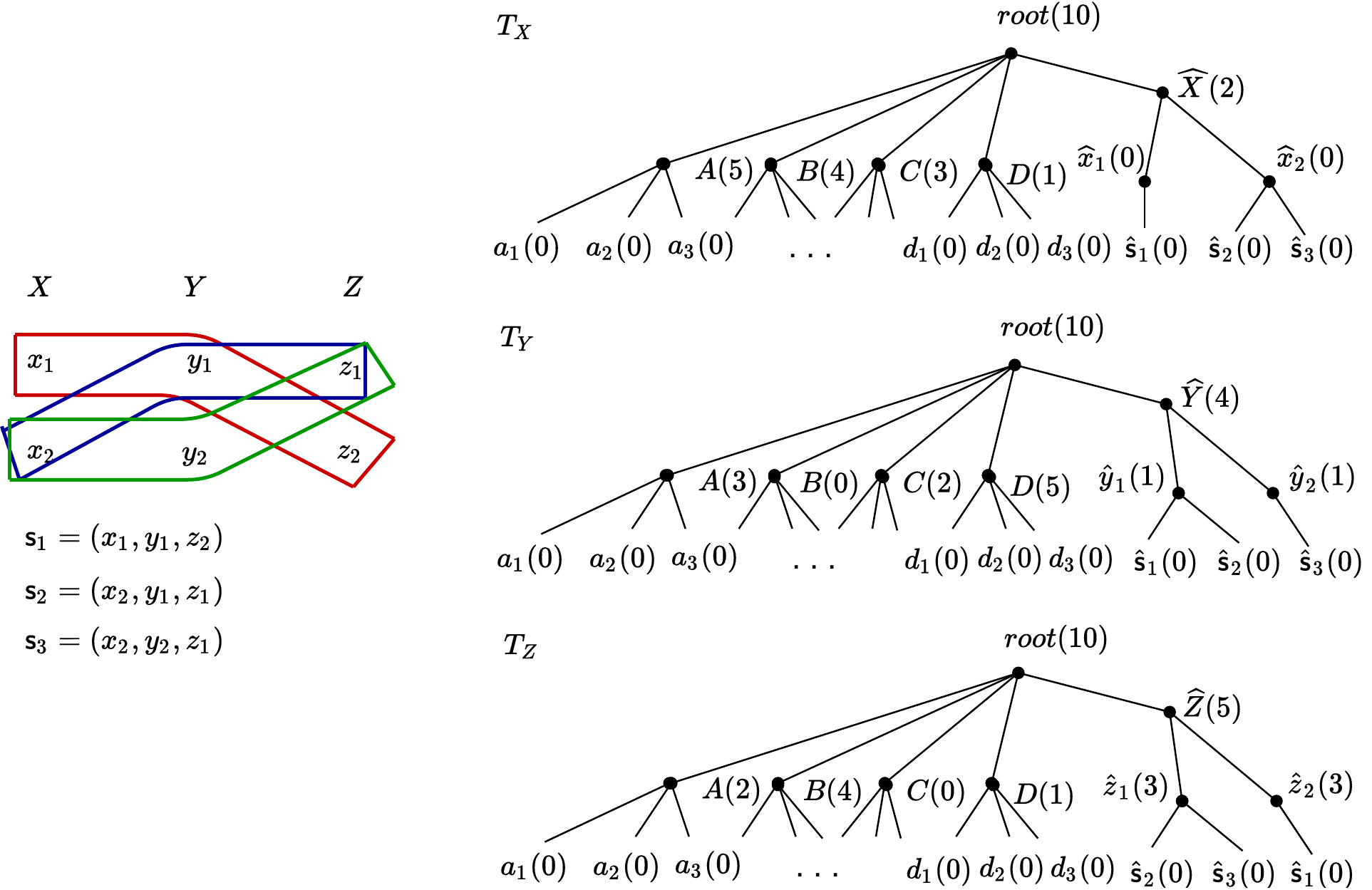}
            \caption{An instance of {\textsc{3D Matching}} problem, and the corresponding instance of \WMIS. The values in parentheses are the height values of tree nodes.}
            \label{fig:3DM_MWI_exp}
        \end{figure}
    
    This finishes setting up all three representing trees (thus also the ultrametrics $U_X, U_Y$ and $U_Z$). Recall that for each ultrametric say $U_X$, the distance $U_X (p, q)$, with $p, q \in P$, corresponds to the height value of the lowest common ancestor (\LCA) of leaves $p$ and $q$. 
    
    In what follows, we will first prove some properties of the constructed ultrametrics. 
    Specifically, consider the auxiliary graph $\mathcal{G}$ constructed for the metric system $(\mathcal{M} = \{U_X, U_Y, U_Z\}; P)$. Recall that each graph node in $\mathcal{G}$ corresponds to a pair of points from $P$, $(p, q)\in P \times P$. For simplicity, we use $\myA$ to represent the set $\{a_1, \ldots, a_m\}$, and similarly for $\myB, \myC$, $\myD$, and $\myR$. 
    Given a graph node $(p,q)$ of $\mathcal{G}$, we say that this pair \emph{splits} if $p$ and $q$ are from two different sets in $\{\myA, \myB, \myC, \myD, \myR\}$ (e.g, $p\in \myA$ and $q\in \myD$). 
    Now consider a triple $\rhat_i = (x', y', z') \in X \times Y \times Z$; we refer to $x'$ (resp, $y'$, $z'$) as the $x$-coordinate (resp. $y$- or $z$-coordinate) of $\rhat_i$. 
    Given a graph node of the auxiliary graph $\mathcal{G}$ of the form $(\rhat_i, \rhat_j)$, we say that this pair \emph{has shared coordinate}, if $\rhat_i \cap \rhat_j \neq \emptyset$. This means that $\rhat_i$ shares either $x$-, $y$- or $z$-coordinate. 
    This is the key lemma to guarantee the correctness of our reduction. The proof of this lemma can be found in Appendix \ref{appendix:lem:nomix}. 
    We remark that the height values of all nodes in the three representing trees are chosen carefully so that the lemma below holds. 
    To compute these height values, we in fact write a computer program testing all possible permutations of heights over internal nodes to make sure all conditions in Lemma \ref{lem:nomix} are satisfied.
    
    \begin{lemma}\label{lem:nomix}
    Consider a graph node $(p, q) \in P \times P$ of the auxiliary graph $\mathcal{G}$. 
    \begin{itemize}
        \item[(i)] If $(p, q)$ splits, then this graph node cannot appear in any directed triangle in the auxiliary graph $\mathcal{G}$.  
        \item[(ii)] Any directed triangle in $\mathcal{G}$ must contain at least one graph node of the form $(\rhat_i, \rhat_j)$ where this pair have shared coordinate. 
        \item[(iii)] If $(p, q)$ is of the form $(\rhat_i, \rhat_j)$ and this pair has {\bf shared coordinate}, then this graph node $(\rhat_i, \rhat_j)$ must participate in at least one directed triangle. 
    \end{itemize}
    \end{lemma}

    Note that the correctness of the reduction then follows easily from the above key lemma. In particular, we now show that $(X, Y, Z; \matchS)$ has a matching $\Pi\subseteq \matchS$ of size $K$ if and only if the metric system $(\{U_X, U_Y, U_Z\}, P)$ has a consensus subset of size $4m + K$. 
    
    \noindent \textbf{``$\Rightarrow$'' direction:} Suppose $(X, Y, Z; \matchS)$ has a matching $\Pi = \{ \myr_{I_1}, \ldots, \myr_{I_K}\}$ of size $K$. Then we claim that the set 
    \begin{equation*}
        \begin{split}
            P' =& \{a_1, \cdots, a_m, b_1, \cdots, b_m, c_1, \cdots, c_m, d_1, \cdots, d_m, \rhat_{I_1}, \cdots, \rhat_{I_K}\} \\
        =& \myA \cup \myB \cup \myC \cup \myD \cup \{ \rhat_{I_1}, \cdots, \rhat_{I_K}\}
        \end{split}
    \end{equation*}
    forms a consensus subset of $P$ w.r.t. the metric system $(\{U_X, U_Y, U_Z\}; P)$. Specifically, by Claim \ref{claim:3_directed_tri}, we just need to show that the subgraph $\mathcal{G}'$ of the auxiliary graph $\mathcal{G}$ spanned by nodes coming from $P'\times P'$ contains no directed triangle. 
    As $\Pi$ is a valid matching, no two $\rhat_{I_i}$ and $\rhat_{I_j}$ can have shared coordinates. 
    It then follows from Lemma \ref{lem:nomix} (ii) that there cannot be any directed triangle in the subgraph $\mathcal{G}'$. 

    \noindent \textbf{``$\Leftarrow$'' direction:}  Suppose we have a consensus subset $P' \subset P$ for the metric system $(\{U_X, U_Y, U_Z\}; P)$ such that $|P'| = 4m+K$. 
    First, consider $P' \cap \myR = \{\rhat_{J_1}, \ldots, \rhat_{J_s} \}$. We know that the subgraph $\mathcal{G}'$ spanned by nodes from $P' \times P'$ contains no directed triangle. By Lemma \ref{lem:nomix} (iii), it then follows that no two $\rhat_{J_i}, \rhat_{J_j}$, $i, j\in [1, s]$, could have shared coordinate. In other words, the set $\{\rhat_{J_1}, \ldots, \rhat_{J_s}\}$ forms a valid 3D matching for $(X, Y, Z; \matchS)$ of size $s$. 
    On the other hand, we know that $|P' \setminus \myR| \le 4m$ (as the largest possible choise for $P' \setminus \myR$ is $\myA\cup \myB \cup \myC \cup \myD$). Since $|P'| = 4m+K$, it then follows that $s \ge |K|$, and thus there exists a 3D matching of $(X, Y, Z; \matchS)$ of size at least $|K|$. 

    As \threeDimensionalMatching is NP-complete, it then follows that the decision problem of \WMIS is NP-complete. This finishes the proof of Theorem \ref{thm:MWI_3_ultra_NP}.    
    \end{proof}
        
    The following theorem is an implication from the previous proof. Similarly as \SMIS, we also provide a direct proof for arbitrary metrics in Appendix \ref{app:prove_thm_weakNP_3metric}. In the proof, we construct a size-preserving reduction from \minVertexCover which again leads to an $(2 - \epsilon)$-inapproximability result. 
        
    \begin{theorem} \label{thm:MWI_arbitrary_3_mets}
          Given a metric system $(\mathcal{M}; X)$, where $\mathcal{M} = \{\Metric_1, \Metric_2, \Metric_3\}$ contains three arbitrary metrics, the decision version of \WMIS is NP-complete.
        
        Furthermore, \WMIS with 3 metrics is Unique Games-hard to approximate within a factor $2 - \epsilon$ for an arbitrarily small positive constant $\epsilon>0$. 
    \end{theorem}
    
    
    %
    
    Finally, there is a simple 6-approximation algorithm with running time $O(n^6)$: 
    Specifically, given a metric system $(\mathcal{M}; X)$ with $n = |X|$, we first build auxiliary graph $\mathcal{G}$ as described earlier in $O(kn^4)$ time.
    We want to construct an outlier set $S \subset X$ so that it ``hits'' all directed triangles in $\mathcal{G}$: note that this will then guarantee that $X \setminus S$ is a consensus set w.r.t. $\mathcal{M}$. 
    To this end, we simply enumerate all directed triangles in $\mathcal{G}$ in $O(n^6)$ time. We then initialize $S = \emptyset$ and go through the list of directed triangles one by one. 
    For each directed triangle $\Delta w_1 w_2 w_3$ (where each $w_i$ is a pair of points in $X$), if $S$ does not intersect with any of the points included in $w_1\cup w_2 \cup w_3 \subset X$, then we simply add all these points (at most 6 distinct points) to $S$. 
    Otherwise, this triangle is already ``hit'' by $S$ and we do nothing. 
    Let $S^*$ be the minimum weakly inconsistent set (i,e, the optimal solution for \WMIS). 
    It is easy to see that $|S^*| \ge |S|/6$, as all the 6-tuples (at most 6) we ever added to $S$ are all disjoint, and for each such 6-tuple, $S^*$ must contain at least one point from it. 
    Furthermore, it is also easy to see that after removing all points $S$, the resulting auxiliary graph restricted to only pairs not containing points in $S$ is free of directed triangles, and thus free of directed cycles by Lemma \ref{lem:validrelationset} and Claim \ref{claim:3_directed_tri}. Hence $X\setminus S$ is a consensus set w.r.t. $\mathcal{M}$, and $|S| \le 6|S^*|$. Hence: 
    \begin{theorem} \label{thm:MWI_6_approx}
    There is a 6-approximation algorithm for the \WMIS problem that runs in $O(kn^4 + n^6)$ time for $k$ metrics defined on a point set of size $n$. 
    \end{theorem}

\section{Conclusion}
        In this paper, we proposed to study the maximum ordinal consensus problem over a set of input metrics. We developed two concepts of ``consistency'' that only rely on ordinal information of pairwise distances. We proved several hardness results for both definitions with different input metrics. We also developed constant-factor approximation algorithms for the minimum inconsistent set problem under both definitions.
        
        There are still some open directions for future work. 
        For example, can we close the gap between the inapproximability and the approximation algorithm we developed? Can we improve the time complexity, especially for the \WMIS problem? Can we find better approximation algorithms for special cases such as Euclidean metrics or ultrametrics? We also note that the current approximation algorithms target the minimum (ordinally) inconsistent subset problems -- how about the dual maximum ordinal consensus problem?  
        
\section{Acknowledgment}
    This work is partially supported by National Science Foundation (NSF) under grant IIS-1815697, as well as National Institute of Health (NIH) under grant R01EB022899.



\newpage
	\bibliography{Reference}

\newpage

\appendix 

\section{Missing Details} \label{app:missing_hardness_proofs}

\subsection{An Ultrametric and its Representing Tree} \label{app:ultra_representingTree}
Figure \ref{fig:ultra_tree} shows an example of an ultrametric and its representing tree. Values on internal nodes are their heights, and all leaf nodes are with height 0. The distance between any two leaf nodes is the height function value of their \LCA. For example, the distance from $A$ to $E$ is 3.
    
    \begin{figure}[ht]
        \centering 
        \includegraphics[width=0.8\textwidth]{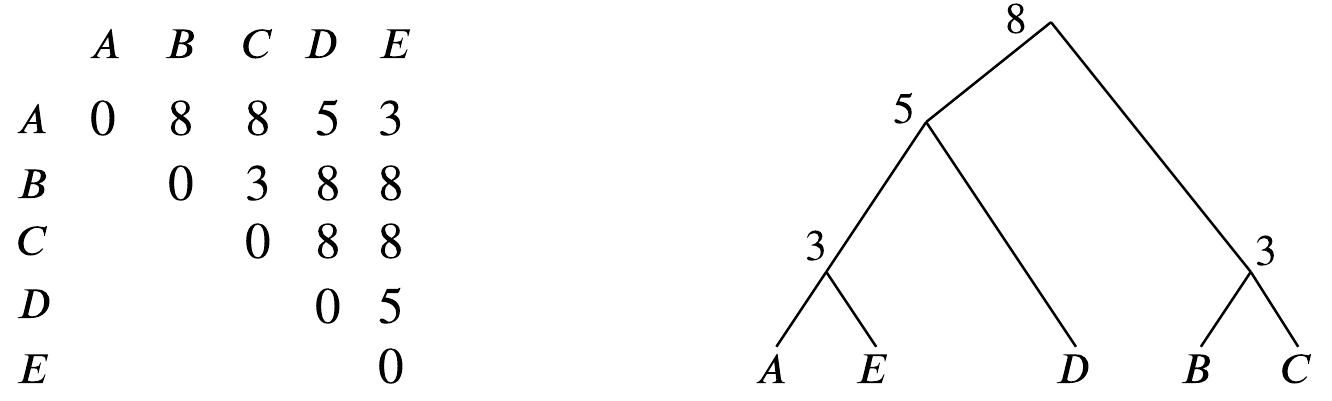}
        \caption{An ultrametric and its corresponding representing tree.}
        \label{fig:ultra_tree}
    \end{figure}

    \subsection{Proof of Theorem \ref{thm:MSI_2_line_mets}}
    \label{app:MSI_2_line_mets_NP}
    
    \begin{proof}
        We prove this theorem via a reduction from the \minVertexCover problem. 
        
        \noindent \textbf{Description of the reduction.}
        Suppose we are given an instance $G = (V, E)$ of \minVertexCover, where $V = \{v_1, ..., v_n\}$ and $E = \{e_1, ..., e_m\}$. A vertex cover $V' \subset V$ is such that every edge in $E$ has at least one endpoint in $V'$. The decision version of the \minVertexCover problem is that, given $G=(V, E)$ and an integer $K$, is there a vertex cover $V' \subset V$ such that $|V'|=K$?
        From the instance of \minVertexCover, we construct an instance of the \SMIS problem $(\mathcal{M} = \{\Metric_1, \Metric_2\}; X)$, where $\Metric_1, \Metric_2$ are two line metrics. Here $X = \{\vhat_1, ..., \vhat_n, r_{e_1, l}, r_{e_1, r}, ..., r_{e_m, l}, r_{e_m, r}\}$, and $|X| = n + 2m$. The $r_{e_i, l}$ and $r_{e_i, r}$ are two ``pivots'' for $i$th edge $e_i$.
        
        Since all these points are on real line $\R$, we will use $``[\cdot]_1"$ and $``[\cdot]_2"$ to represent coordinates. When the coordinate is the same for both metrics or there is no ambiguity, we will omit the subscript and use $``[\cdot]"$ for the coordinate. For example, $[\vhat_1]_1$ is the coordinate of point $\vhat_1$ in $\Metric_1$. 
        Instead of constructing $\Metric_1, \Metric_2$ directly, we construct the coordinates of each point as follows (see Figure \ref{fig:strong_2real_line} for an example): 
        \begin{enumerate}
            \item For both $\Metric_1, \Metric_2$, $[\vhat_i] = 2^{i-1}$. 
            \item For both $\Metric_1, \Metric_2$, the coordinate of $[r_{e_j, l}] = 2^{n + 2(j-1)}$.
            \item For $r_{e_j, r}, 1 \leq j \leq m$, and $e_j = (v_a, v_b), b > a$, its coordinate $[r_{e_j, r}]_1$ in $\Metric_1$ is: $[r_{e_j, r}]_1 = [r_{e_j, l}]_1 + ([\vhat_b]_1 - [\vhat_a]_1) + \epsilon = 2^{n + 2(j-1)} + (2^{b-1} - 2^{a-1}) + \epsilon$. 
            Similarly, its coordinate $[r_{e_j, r}]_2$ in $\Metric_2$ is: $[r_{e_j, r}]_2 = [r_{e_j, l}]_2 + ([\vhat_b]_2 - [\vhat_a]_2) - \epsilon = 2^{n + 2(j-1)} + (2^{b-1} - 2^{a-1}) - \epsilon$.
            The intuition is the distance from $r_{e_j, l}$ to $r_{e_j, r}$ is close to the distance from $v_a$ to $v_b$, but orders are different in two metrics.
        
        \end{enumerate}
        
        \begin{figure}[ht]
            \centering 
            \includegraphics[width=0.75\textwidth]{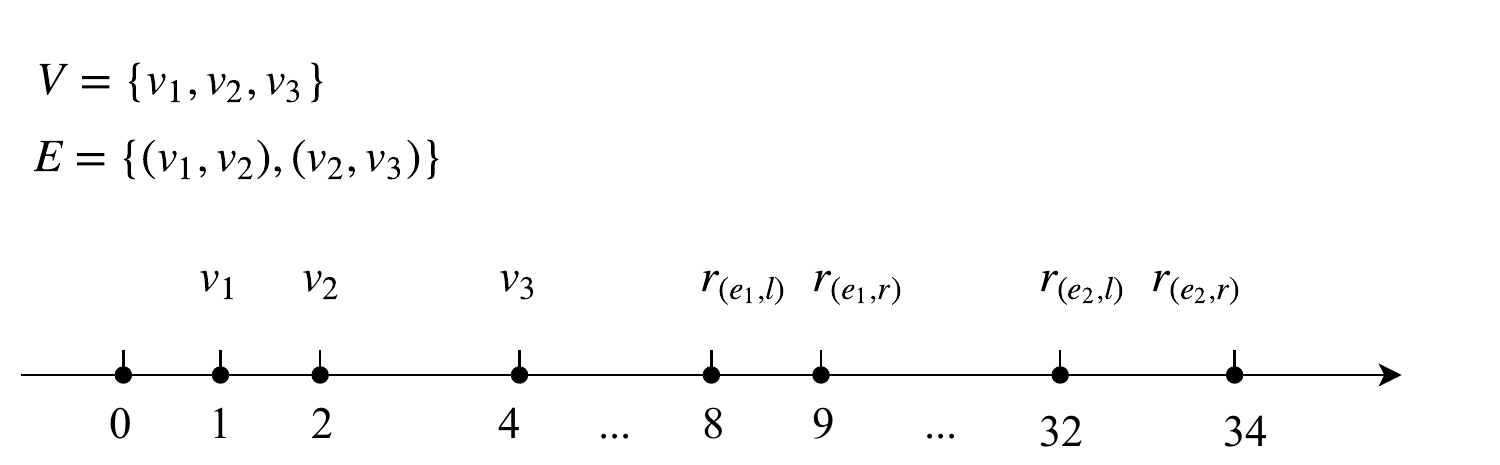}
            \caption{Constructed line metrics given an simple $G$.}
            \label{fig:strong_2real_line}
        \end{figure}
        
        \noindent \textbf{Check comparable pairwise distances.}
        The distance between any pair of nodes under metrics $\Metric_1$ and $\Metric_2$ can be easily calculated based on coordinates $[\cdot]_1$ and $[\cdot]_2$. 
        We now show that any conflict quartet of the above constructed instance must have the form of $\{\vhat_i, \vhat_j, r_{e, l}, r_{e, r}\}$, where $e = (v_i, v_j) \in E$. 
        We first consider the necessary conditions for two pairs $(a, b)$ and $(c, d)$ to cause a conflict. W.l.o.g., we can assume that $[b] > [a]$ and $[d] > [c]$ in both metrics.
        \begin{itemize}
            \item [(1)] To cause a conflict, the distances $[b] - [a]$ and $[d] - [c]$ should be \emph{close} so that the relation between $\{[b]_1 - [a]_1, [d]_1 - [c]_1\}$ and $\{[b]_2 - [a]_2, [d]_2 - [c]_2\}$ can be different. One can notice that the distance can be distorted by at most $2 \cdot \epsilon$ from one metric to the other, thus in a more formal way, \emph{close} means ``differ by at most $4 \cdot \epsilon"$.
            \item [(2)] Another necessary condition for $(a, b)$ and $(c, d)$ to cause a conflict is that $[b] - [a]$ is not identical to $[d] - [c]$ in any one of the two metrics. For example, assume that $[b]_1 - [a]_1 = [d]_1 - [c]_1$ (equivalently $[b]_1 + [c]_1 = [a]_1 + [d]_1)$ and $[b]_1 + [c]_1 = I + t \cdot \epsilon = [a]_1 + [d]_1$. Here $I$ and $t$ are some integers. When comes to coordinates in the second metric, one can show that $[b]_2 + [c]_2 = I - t \cdot \epsilon = [a]_2 + [d]_2$ as only the polarity of $\epsilon$ term is changed across two metrics.
        \end{itemize}
        Combining those two conditions, we define two pairs are \emph{comparable} if they are \emph{different} and \emph{differ by at most $4 \cdot \epsilon$}. A conflict is possible only when two pairs in the quartet are comparable.
        We then iterate over all pairs and check whether there is any other pair with a comparable pairwise distance.
        We will also use the \emph{leading term} (the largest power of two) of coordinates to find comparable pairs. The leading term of $[\vhat_i]s$ and $[r_{e, l}]$s are of different scales, and $[r_{e, r}]$ has the same leading term of $[r_{e, l}]$. Moreover, if $[b]$ has a larger leading term than $[a]$, then $[b] \ge 2[a]$.
        
        The necessary condition of two pairs $(a, b)$ and $(c, d)$ (where we assume $[b] > [a]$, $[d] > [c]$ and $[d] \ge [b]$) having comparable distances is either 1). nodes $d$ and $b$ have the leading terms of the same scale. or 2). nodes $c$ and $d$ have the leading term of the same scale. The reason is (assuming $[d]$ has a larger leading term than $[b]$ and $[c]$) that $[d] >= 2[c]$ and $2[b] <= [d]$ implies $[d] - [c] \ge [d]/2 \ge [b] \ge [b] - [a] + 1$ if none of the conditions are satisfied. In either case where the necessary condition is satisfied, there must be two nodes of the form $(r_{e, l}, r_{e, r})$ for some edge $e$. This observation largely reduced the possibilities of comparable pairs. All possible cases are considered below (we use ``$x, y$'' to denote unknown elements and ``$\approx$'' to denote comparable distances):
        
        \begin{enumerate}
            \item $(x, r_{e, l})$ and $(y, r_{e, r})$ for some edge $e$. We have $[r_{e, r}] - [y] \approx [r_{e, l}] - [x] \Rightarrow [y] - [x] \approx [r_{e, r}] - [r_{e, l}]$. Assume $e = (v_i, v_j)$, and we have $[y] - [x] \approx 2^{j-1} - 2^{i-1}$, this is true only when $y = \vhat_j$ and $x = \vhat_i$ by checking the leading term of $y$.
    
            \item $(x, y)$ and $(r_{e, l}, r_{e, r})$. By similar calculation, we should easily get $y = \vhat_j$ and $x = \vhat_i$ assuming $e = (v_i, v_j)$.
        \end{enumerate}
        
        And indeed, $\{\vhat_i, \vhat_j, r_{e, l}, r_{e, r}\}$ is a conflict quartet (where $e = (v_i, v_j)$) since $[\vhat_j]_1 - [\vhat_i]_1 < [r_{e, r}]_1 - [r_{e, l}]_1$ but $[\vhat_j]_2 - [\vhat_i]_2 > [r_{e, r}]_2 - [r_{e, l}]_2$. Now, we can conclude that the quartets in conflict set $\mathcal{C}$ have one-to-one correspondence to edges in $E$, i.e., $\mathcal{C} = \{(\vhat_i, \vhat_j, r_{e, l}, r_{e, r} | e = (v_i, v_j) \in E\}$. Now it is not hard to prove there is a size-preserving reduction from \minVertexCover to \SMIS (with two line metrics).
                
        \noindent \textbf{There is a vertex cover of size $K$ for $G$ $\Rightarrow$ there is an inconsistent set of size $K$ for $(\mathcal{M}=\{\Metric_1, \Metric_2\}; X)$:}
        Assume there is a vertex cover of $G$ with size $K$ denoted by $V' = \{v_1, \cdots, v_K\}$. Then the corresponding set $\Shat = \{\vhat_1, \cdots, \vhat_K\}$ is a solution for $(\mathcal{M}=\{\Metric_1, \Metric_2\}; X)$, since this set covers all edges and thus all possible conflict quartets. The remaining nodes will have no conflict quartet.
        
        \noindent \textbf{$``\Leftarrow"$ direction:}
        Conversely, assume there is an optimal inconsistent set with size $K$. It is clear that there is always an optimal solution to $(\mathcal{M}=\{\Metric_1, \Metric_2\}; X)$ which only includes $\vhat_i$s. The reason is that including $\vhat_i$ can potentially cover multiple conflict quartets, but $r_{e, l}$ (or $r_{e, r}$) will only cover at most one conflict. For any solution containing $r_{e, l}$ (or $r_{e, r}$), including one of the endpoints of $e$ will also cover the conflict (covered by $r_{e, l}$). And it can potentially cover other conflicts and thus reduce the size of the solution. Assume one optimal solution is $S^*=\{\vhat_1, \cdots, \vhat_K\}$ which covers all conflict quartets, then the corresponding node set $\{v_1, \cdots, v_K\}$ is a vertex cover covering all edges due to the one-to-one correspondence between edges in $E$ and quartets in $\mathcal{C}$.
    \end{proof}

    \subsection{Proof of Theorem \ref{thm:MSI_2_ultra_NP}}
    \label{app:MSI_2_ultra_NP}
    \begin{proof}
        We prove this theorem via a reduction from the \maxTwoSAT problem.
        
        \noindent \textbf{Description of reduction.}
        Suppose we are given an instance $(C, X)$ of the \maxTwoSAT problem, where $C = \{c_1, ..., c_{m_1}, c_{m_1+1}, c_{m_1+m_2}\}$ is a set of clauses. There are $m_1$ clauses with one literal and $m_2$ clauses with two literals. And $X = \{x_1, ..., x_n\}$ is a set of variables. The decision version of the \maxTwoSAT problem is that, given $(C, X)$ and an integer $K$, is there an assignment to variables such that $K$ clauses are satisfied?
        
        From the instance of \maxTwoSAT problem, we construct the following instance of \SMIS $(\mathcal{M} = \{\Metric_1, \Metric_2\}; \widehat{X})$ where $\Metric_1$ and $\Metric_2$ are ultrametrics. In particular, we will set up each metric via its representing tree as introduced in Definition \ref{def:ultrametric}. As we will see below, both representing trees of $\Metric_1, \Metric_2$ will have the same structure; while the only difference is the height functions over internal nodes.
        
        First, we will describe the tree structure $T$ (which will be common for both $\Metric_1$ and $\Metric_2$). 
        The root of $T$ has $m_1 + m_2 + 1$ children, $\{B, \chat_1, ..., \chat_{m_1 + m_2}\}$, where $\chat_i$s correspond to $i$th clauses in $C$. The node $B$ has $|\Child(B)| = 2m_1 + 4m_2 + 1$ children denoted by $b_1, ..., b_{2m_1 + 4m_2 + 1}$ (here for simplicity, we use $\Child(B)$ to denote the children set of $B$). The size of $\Child(B)$ is designed to be larger than the total number of literals (allowing duplicates) appeared in the clause set $C$.
        We show later that one can assume a maximum consensus always contain all nodes from $\Child(B)$.
        Each $\chat_i$ has at most two children $l_{i1}, l_{i2} \in \{x_1, ..., x_n, \bar{x_1}, ..., \bar{x_n}\}$ corresponding to the literals in clause $c_i$, and each $l$ has two leaf nodes $d$s as its children. For example, if clause $c_i$ has two literals, then $\chat_i$ will have two children $l_{i1}, l_{i2}$; And $l_{i1}$ has two children (leaves) $d_{i1}, d_{i2}$, $l_{i2}$ has two children (leaves) $d_{i3}, d_{i4}$. If clause $c_i$ has only one literal, then $\chat_i$ will have one child $l_{i1}$ which has two children (leaves) $d_{i1}$, $d_{i2}$. See an illustration of the tree structure $T$ in Figure \ref{fig:MSI_2_ultra_structure}. 
        
        Recall that the leaf set of this representing tree $T$ corresponds to the note set of the metric. Hence, in the constructed \SMIS instance, we have that the note set is $\widehat{X} = \{b_1, ..., b_{2m_1 + 4m_2 + 1}, d_{11}, d_{12}, \cdots, d_{m_11}, d_{m_12}, \cdots, d_{(m_1+m_2)1}, \\ d_{(m_1+m_2)2}, d_{(m_1+m_2)3}, d_{(m_1+m_2)4}\}$ with size $2m_1 + 4m_2 + 1 + 2m_1 + 4m_2 = 8m_2 + 4m_1 + 1$.
        
        \begin{figure}[ht]
            \centering 
            \includegraphics[width=1\textwidth]{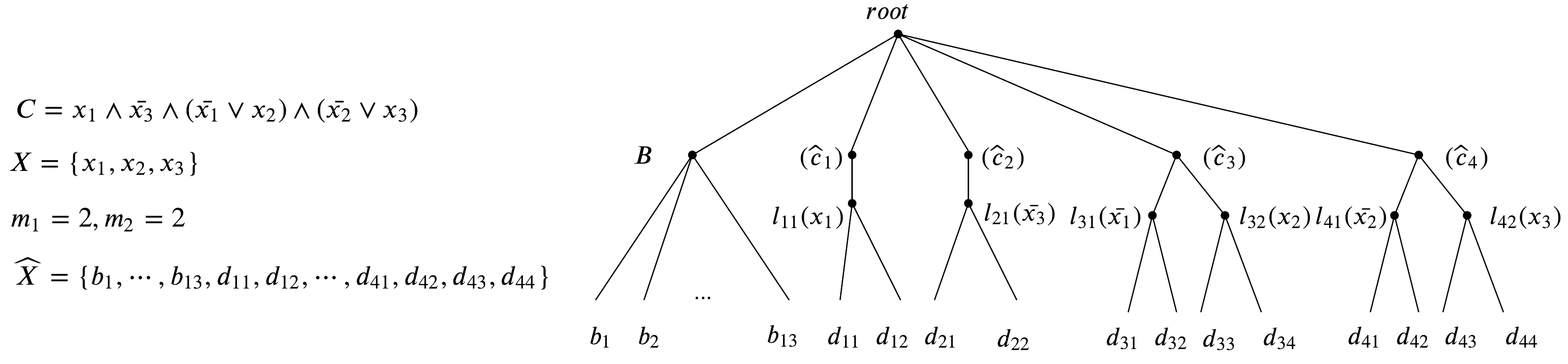}
            \caption{Given an instance of \maxTwoSAT with $C = x_1 \wedge \bar{x_3} \wedge (\bar{x_1} \vee x_2) \wedge (\bar{x_2} \vee x_3)$, the representing tree structure of the constructed ultrametrics is shown on the right. $c_i$ corresponds to the $i$th clause in $C$. $l_{i1}$ (and $l_{i2}$) are its literals. The corresponding literals are shown in the bracket.}
            \label{fig:MSI_2_ultra_structure}
        \end{figure}
        
        Now we equip this tree structure $T$ with two height functions $h_1$ and $h_2$ mapping internal nodes of $T$ to real values. It generates metrics $\Metric_1$ and $\Metric_2$, respectively. 
        In particular, recall that given any two leaves $z$ and $z'$, their distance $\Metric_i(z, z') = h_i(\LCA(z, z'))$. 
        We set up $h_1, h_2$ so that $h_1(x_1) < h_1(\bar{x_1}) < h_1(x_2) < h_1(\bar{x_2}) < ... < h_1(x_n) < h_1(\bar{x_n}) < h_1(B) < h_1(\chat_1) < ... < h_1(\chat_{m_1+m_2})$ and  $h_2(\bar{x_1}) < h_2(x_1) < h_2(\bar{x_2}) < h_2(x_2) < ... < h_2(\bar{x_n}) < h_2(x_n) < h_2(\chat_1) < ... < h_2(\chat_{m_1+m_2}) < h_2(B)$. 
        Note that the precise value of each height does not matter as only the ordering of pairwise distances matters. We use $(T_1, h_1)$ and $(T_2, h_2)$ ($T_1$ and $T_2$ have the same tree structure as $T$) to denote the two representing trees for $\Metric_1$ and $\Metric_2$, respectively. 
        
        To see how the maximum consensus problem for $(\{\Metric_1, \Metric_2\}, \widehat{X})$ relates to the maximum satisfiability of $(C, X)$, note the following: (1) Intuitively, the height functions $h_1$ and $h_2$ guarantee that, for a consensus $\widehat{X}' \subset \widehat{X}$, we cannot include both children of $x_i$ and $\bar{x}_i$, the heights of $x_i$ and $\bar{x}_i$ have opposite orders and thus inconsistent w.r.t. $\Metric_1$ and $\Metric_2$. For instance, in the example of Figure \ref{fig:MSI_2_ultra_structure}, one cannot include $d_{11}, d_{12}$ (children of $l_{11} (x_1)$) and $d_{31}, d_{32}$ (children of $l_{31} (\bar{x_1})$) in $\widehat{X}'$.
        (2) The heights of the internal node $B$ guarantee that for each $\chat_i$, one can only choose either one of its leaf nodes or two of its leaf nodes which are the children of the same literal (See Figure \ref{fig:MSI_2ultra_second_direction}). For the instance shown in Figure \ref{fig:MSI_2_ultra_structure}, we can include both $d_{41}$ and $d_{42}$ in $\widehat{X}'$ since they are the children of the same literal $l_{41}$. However, one cannot include both $d_{41}$ and $d_{43}$ in $\widehat{X}'$ because $h_1(B) < h_1(\chat_4)$ and $h_2(\chat_4) < h_2(B)$. 
        
        Basically, for the construction, the consensus can include two leaf nodes when the corresponding clause is satisfied; Otherwise it can only include one leaf node. 
        Below we show that there is an assignment of $X$ satisfying $K$ clauses in $C$ \emph{if and only if} there is a consensus of $\widehat{X}' \subset \widehat{X}$ of cardinality $|\widehat{X}'| = |\Child(B)|+m_1+m_2+K$. 
        
        \noindent \textbf{$``\Rightarrow''$ direction:}
        
        Assume there is an assignment to variables of $\{C, X\}$ such that $K$ clauses are satisfied. Then we have a consensus of size $|\Child(B)| + m_1 + m_2 + K$ as follows. 
        We first include all $b_i$s. For each satisfied clause, we randomly selected one of its literal with true value and keep both of its children. For the clauses not satisfied, we keep any one of its four leaves. The selected node set is strongly consistent, because (1): there are no leaf nodes from different literals of a clause; 2). we do not keep both leaf nodes from two literals (say, $x_i$ and $\bar{x_i}$) of the same variable as we only keep both leaves of true literals. The size of the set is $|\Child(B)| + m_1 + m_2 + K$. 
        
        More specifically, w.l.o.g., assume nodes in $\widehat{X}'$ are $\{b_1, \cdots, b_{2m_1 + 4m_2 +1}, \cdots, \\ d_{11}, d_{12}, \cdots, \\d_{K1}, d_{K2}, \cdots, d_{(m_1+m_2)1}\}$ (See Figure \ref{fig:MSI_2ultra_first_direction}). We can also assume that the corresponding literals for $l_{11}, \cdots, l_{K1}$ are $x_1, \cdots, x_K$. Then all pairwise distances in $\Metric_1$ form the set $\{h_1(x_1), h_1(x_2), h_1(x_K), h_1(B), h_1(root)\}$ (and $\{h_2(x_1), \cdots, \\ h_2(x_K), h_2(B), h_2(root)\}$ for $\Metric_2$). One can check the order of those heights are the same in two metrics, and thus $\widehat{X}'$ is a consensus. 
        
        \begin{figure}[ht]
            \centering 
            \includegraphics[width=0.6\textwidth]{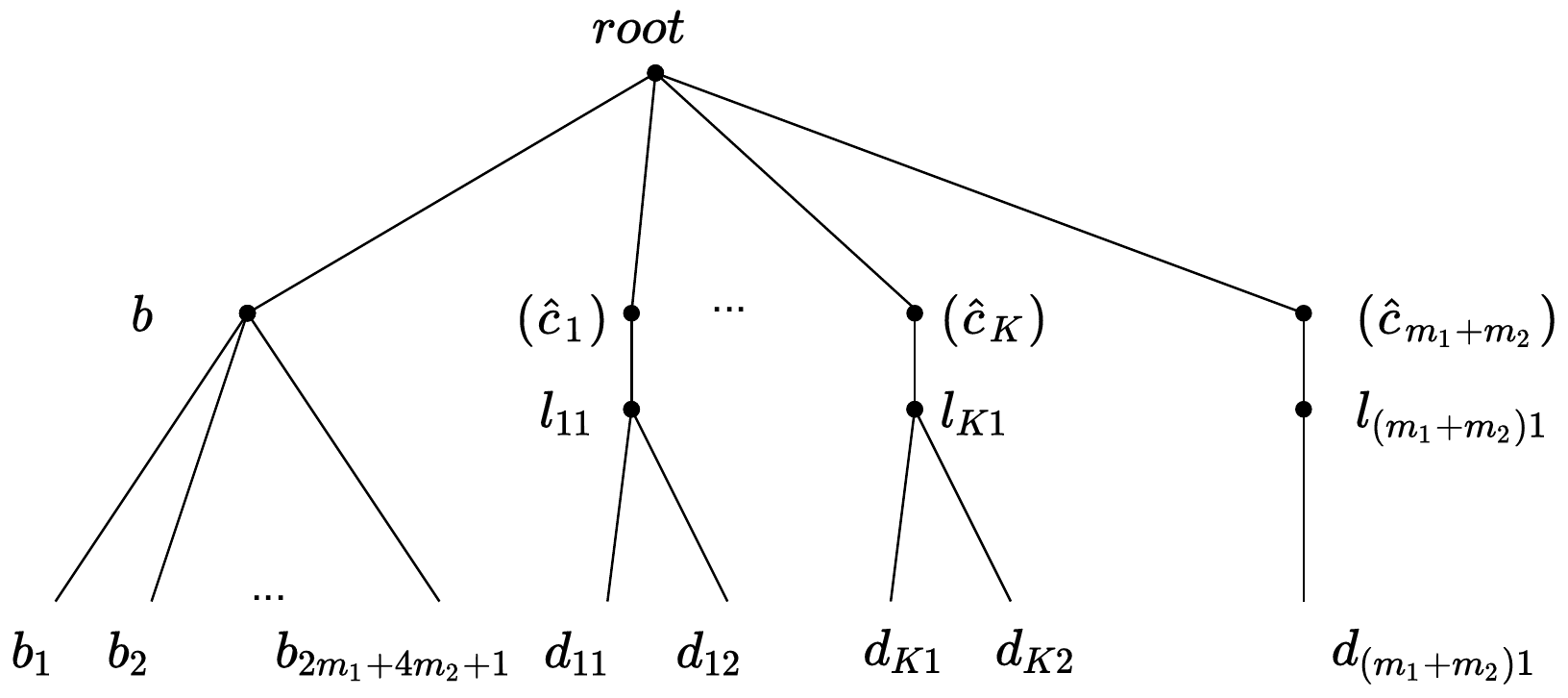}
            \caption{An example of leaf nodes included in $\widehat{X}'$. For each satisfied clause $\{c_1, \cdots, c_K\}$, we can include two out of four leaves (these two leaves have to be the children of the same literal). For unsatisfied clauses, we can keep one leaf. This leads to a consensus with $|\Child(B)| + m_1 + m_2 + K$ nodes. The internal nodes with only one child can be ignored (e.g., $l_{(m_1+m_2)1}$).}
            \label{fig:MSI_2ultra_first_direction}
        \end{figure}
        
        \noindent \textbf{$``\Leftarrow"$ direction:}
        Assume there is maximum consensus of $(\mathcal{M}=\{\Metric_1, \Metric_2\}; \widehat{X})$ with $|\Child(B)| + m_1 + m_2 + K$ nodes. Then we want to show that there is an assignment satisfying at least $K$ clauses based on the following observations.
        
        \begin{enumerate}
            \item There is always an optimal maximum consensus that keeps all $b_i$s. Simply including all $b_i$s will produce a consensus with $|\Child(B)| = 2m_1 + 4m_2 + 1$ nodes. If one optimal solution does not include all $b_i$s, it must have removed at least $2m_1 + 4m_2$ $b_i$s to avoid any conflict. In this case, there are only at most $2m_1 + 4m_2 + 1$ nodes left.
            \item For the restricted tree on the nodes from the maximum consensus (including all $b_i$s), there are only two possible options for a subtree rooted at $\chat_i$s (See Figure \ref{fig:MSI_2ultra_second_direction}). The reason is that we cannot keep leaf nodes of both literals (e.g., we cannot include both $d_{31}$ and $d_{33}$ for the case in Figure \ref{fig:MSI_2_ultra_structure}), which will cause a conflict since $h_1(B) < h_1(\chat_i)$ but $h_2(\chat_i) < h_2(B)$.
            \item For the literal with both children included, there is no conflict. Therefore, e.g., if $x_i$ has both children included, then $\bar{x_i}$ cannot. This is because the heights of $x_i$ and $\bar{x_i}$ are not consistent, their four leaf nodes form a conflict quartet.
        \end{enumerate}
        
        \begin{figure}[ht]
            \centering 
            \includegraphics[width=0.5\textwidth]{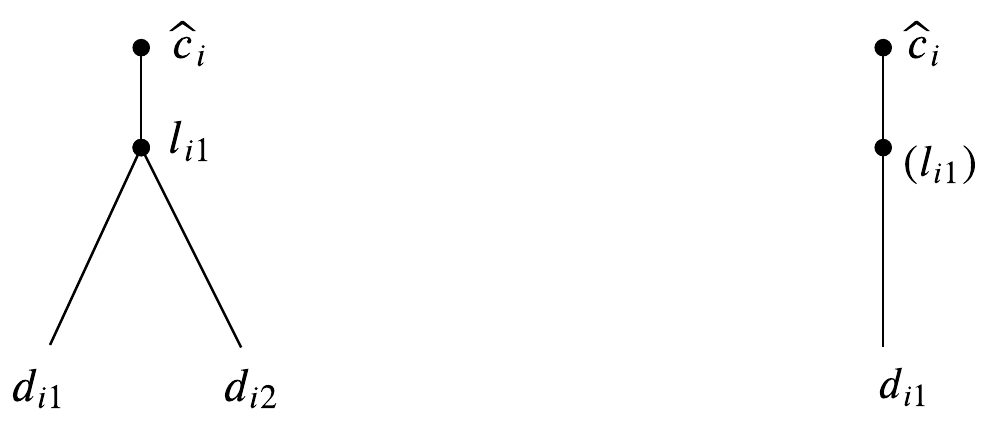}
            \caption{Two cases of subtrees rooted at $\chat_i$. The indexes of literals or leaves are not important, so they are omitted.}
            \label{fig:MSI_2ultra_second_direction}
        \end{figure}
        
        Based on these observations, there is an optimal solution having exactly $K$ subtrees rooted at $\chat_i$s with two leaves selected.  The two leaves must be the children of one literal. Also, there will not be any contradiction from those $K$ literals with two children selected (i.e., $x_i$, $\bar{x_i}$ cannot both have two children selected) 
        We can assign true value to these literals, and the $K$ corresponding clauses are satisfied.
    \end{proof}

\subsection{Proof of Theorem \ref{thm:strong_NP_2metrics}} \label{app:prove_thm_strongNP_2metric}

    \begin{proof}
        We prove this theorem via a reduction from \minVertexCover. 
        
        \noindent \textbf{Description of reduction.}
        Given an instance of \minVertexCover, $G = (V, E)$, where $V = \{v_1, v_2, ..., v_n \}$, $E = \{e_1, e_2, ..., e_m\}$. We construct an instance $(\mathcal{M}; X)$ of \SMIS, where $ \mathcal{M} = \{\Metric_1, \Metric_2\}$ ($\Metric_1, \Metric_2$ are two arbitrary metrics) and $X = \{r_1, ..., r_n, \vhat_1, ..., \vhat_n\}$. Here $\vhat_1, ..., \vhat_n$ correspond to nodes $\{v_1, v_2, ..., v_n \}$ in $V$.
        
        For these two metrics, $\Metric_1(r_i, \vhat_j) = \Metric_2(r_i, \vhat_j) = 1$, and $\Metric_1(r_i, r_j) = \Metric_2(r_i, r_j) = 1$, for $\forall i, j$. Those pairwise distances are used as standards which will be compared with pairwise distances of $\Metric_1(\vhat_i, \vhat_j)$ and $\Metric_2(\vhat_i, \vhat_j)$, for those pairwise distances, 
        \[
        \Metric_1(\vhat_i, \vhat_j)=\begin{cases}
        1 + \epsilon \text{ if }(v_i, v_j) \in E,\\
        0, \text{ if } i=j, \\
        1 \text{ otherwise.}
        \end{cases}
        \Metric_2(\vhat_i, \vhat_j)=\begin{cases}
        1 - \epsilon \text{ if }(v_i, v_j) \in E,\\
        0, \text{ if } i=j, \\
        1 \text{ otherwise.}
        \end{cases}
        \]
        
        Clearly, for a fixed edge $(v_i, v_j) \in E$, $(\vhat_i, \vhat_j, r_k, r_l), \forall 1 \le k, l \le n$ is a conflict quartet, since $\Metric_1(\vhat_i, \vhat_j) > \Metric_1(r_k, r_l) = 1$ but $\Metric_2(\vhat_i, \vhat_j) < \Metric_2(r_k, r_l) = 1$. We remark that there is an optimal inconsistent set $S^*$ of $(\mathcal{M};X)$ which does not have any point from $\{r_1, ..., r_n\}$. It is because $|S^*| < n$ by removing all $\vhat_i$'s except one; And to cover any conflict quartet, one has to remove at least $n-1$ $r_i$s. 
        
        Now, we want to show that there is an optimal vertex cover of size $K$ if and only if there is an inconsistent subset of size $K$. 
        
        \noindent \textbf{$``\Rightarrow"$ direction:}
        If there is a vertex cover of size $K$ for $G$, denoted by $V' = \{v_1, ..., v_K\}$. Then the corresponding set $\Shat = \{\vhat_1, ..., \vhat_K\}$ is a inconsistent set, because for the remaining nodes $\{\vhat_i \,|\, \vhat_i \in X \backslash \Shat\}$, their pairwise distances are all 1 in both metrics, therefore they form a consensus together with $r_i$s.
        
        \noindent \textbf{$``\Leftarrow"$ direction:}
        Conversely, if there is a minimum inconsistent set with size $K$ (notice that $K$ is always smaller than $n$ and only contains nodes from $\{\vhat_1, ..., \vhat_n\}$), denoted as $S^* = \{\vhat_1, ..., \vhat_K\}$, then there is a vertex cover $\{v_1, ..., v_k\}$. That is because, by removing $\vhat_1, ..., \vhat_K$, the remaining elements form a consensus set. Thus for any edge $(v_i, v_j) \in E$, there is at least one of $\vhat_i, \vhat_j$ in $S^*$.
        
    \end{proof}

    \subsection{Missing Figure in Theorem \ref{thm:MSI_4_approx_alg}}
    \label{app:missing_figure_MSI_4approx_alg}
    
    The following figure shows an illustration for the 4-approximation algorithm for $\SMIS$ problem.
    \begin{figure}[h]
            \centering 
            \includegraphics[width=0.75\textwidth]{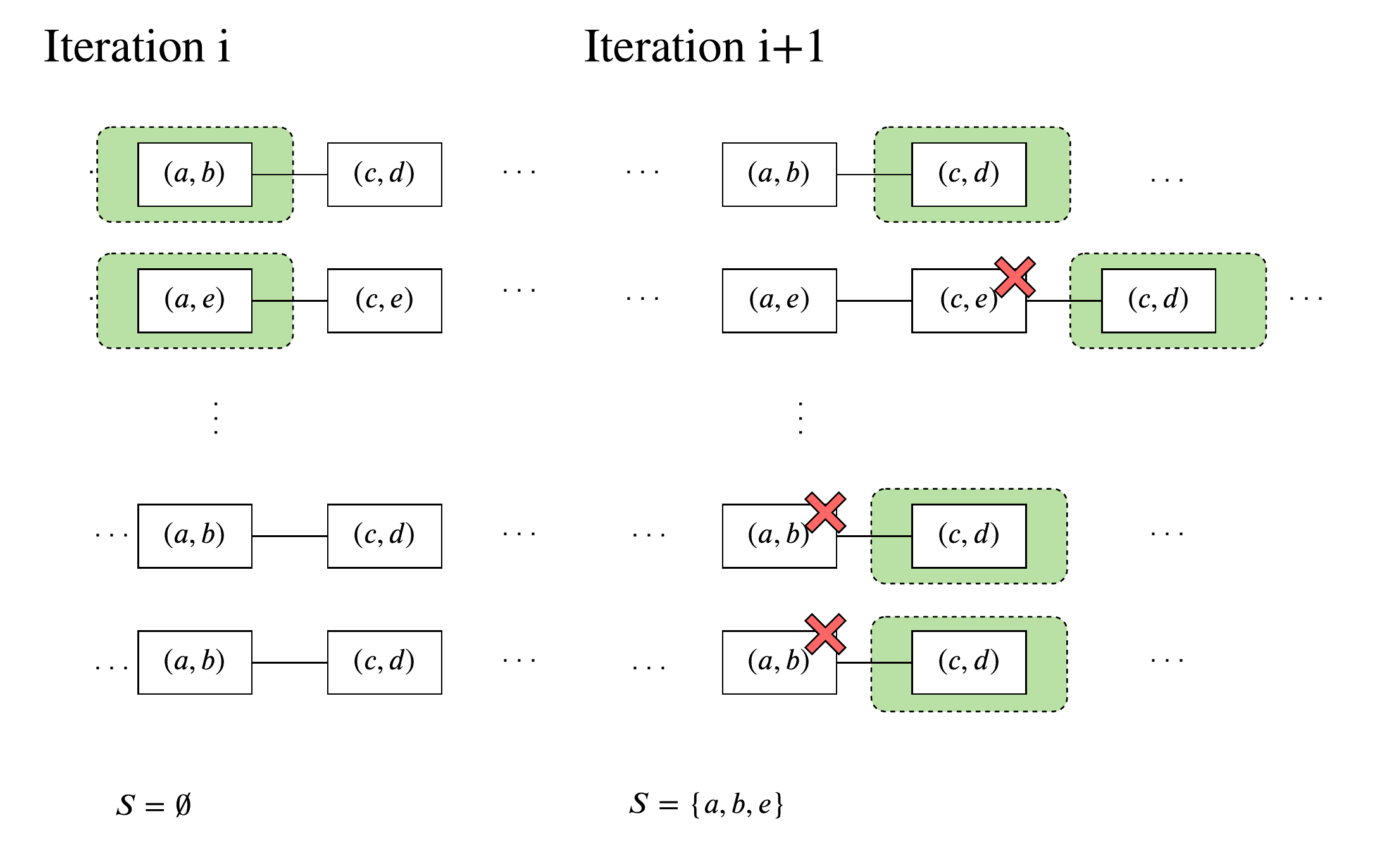}
            \caption{An illustration for our algorithm. In iteration $i$, the $k$ pointers are pointing to the head of those $k$ lists. When testing the second list, there is already a conflict. $\{a, b, e\}$ are included in the inconsistent set $S$. In the beginning of the next iteration $i+1$, it first skips all the edges with nodes in $S$. Edges in green are considered in each iteration.}
            \label{fig:SMIS_4_approx_alg}
        \end{figure}

    \subsection{Proof of Claim \ref{claim:3_directed_tri}}
    \label{app:prove_claim_3_directed_tri}
    \begin{proof}
        Given a mixed graph $G$, assume $U = \{u_0, u_1, ..., u_j\}$ is the directed cycle with smallest length and $j > 2$. For any edge $(u_i, u_{i+1}), 0 \le i \le k$ (addition modulo $k+1$), it is either undirected or from $u_i$ to $u_{i+1}$. W.l.o.g., let $(u_0, u_1)$ be a directed edge. Consider the edge connecting $u_1$ and $u_j$, there are three cases. If the edge connecting $u_1$ and $u_j$ is undirected or from $u_1$ to $u_j$, $\{u_0, u_1, u_j\}$ form a directed triangle. Otherwise (i.e., the edge is from $u_j$ to $u_1$), then $\{u_1, ..., u_j\}$ forms a smaller directed cycle. Both cases contradict with the assumption of smallest length.
    \end{proof}

    \subsection{Missing proof of Theorem \ref{thm:MWI_3_line_metrics}}
    \label{app:prove_MWI_3_line_mets}
    \begin{proof}
        We prove this theorem via a reduction from the \minVertexCover problem. Here we propose a similar construction of Theorem \ref{thm:MSI_2_line_mets}.
        
        \noindent \textbf{Description of the reduction.}
        Suppose we are given a \minVertexCover instance $G = (V, E)$, $V = \{v_1, ..., v_n\}$, $E = \{e_1, ..., e_m\}$. We construct an instance ($\mathcal{M} = \{\Metric_1, \Metric_2, \Metric_3\}; X\}$ of the \WMIS problem, where $\Metric_1, \Metric_2, \Metric_3$ are line metrics and $X = \{\vhat_1, ..., \vhat_n, r_{e_1, l_1}, r_{e_1, l_2}, r_{e_1, r_1}, r_{e_1, r_2}, ..., r_{e_m, l_1}, r_{e_m, l_2}, r_{e_m, r_1}, \\ r_{e_m, r_2}\}$. Here, we also use $``[\cdot]_1"$, $``[\cdot]_2"$, and $``[\cdot]_3"$ to represent the coordinates in $\Metric_1$, $\Metric_2$ and $\Metric_3$, respectively. When the coordinate is the same across all three metrics or there is no ambiguity, we will omit the subscript and use $``[\cdot]"$ instead. 
        
        The coordinates are constructed as follow, and the intuition is that only triples like $\{(\vhat_i, \vhat_j), (r_{e, l_1}, r_{e, r_1}), (r_{e, l_2}, r_{e, r_2})\}$ (where $e = (v_i, v_j)$) can potentially cause a directed triangle (See Figure \ref{fig:weak_3_line} for an example):
        
        \begin{enumerate}
            \item For all of these line metrics, the coordinates $[\vhat_i] = 2^{i-1}$.
            \item For all of these metrics, $[r_{e_j, l_1}] = 2^{n + 4j - 4}$, and $[r_{e_j, l_2}] = 2^{n + 4j - 2}$.
            \item For each edge $e = (v_i, v_j)$, in $\Metric_1$, $[r_{e, r_1}]_1 = [r_{e, l_1}]_1 + [\vhat_i]_1 - [\vhat_j]_1 - 2 \epsilon$, $[r_{e, r_2}]_1 = [r_{e, l_2}]_1 + [\vhat_i]_1 - [\vhat_j]_1 - \epsilon$; in $\Metric_2$, $[r_{e, r_1}]_2 = [r_{e, l_1}]_2 + [\vhat_i]_2 - [\vhat_j]_2 + \epsilon$, $[r_{e, r_2}]_2 = [r_{e, l_2}]_2 + [\vhat_i]_2 - [\vhat_j]_2 + 2 \epsilon$; in $\Metric_3$, $[r_{e, r_1}]_3 = [r_{e, l_1}]_3 + [\vhat_i]_3 - [\vhat_j]_3 + \epsilon$, $[r_{e, r_2}]_3 = [r_{e, l_2}]_3 + [\vhat_i]_3 - [\vhat_j]_3 - \epsilon$.
        \end{enumerate}
        
        \begin{figure}[ht]
            \centering 
            \includegraphics[width=0.75\textwidth]{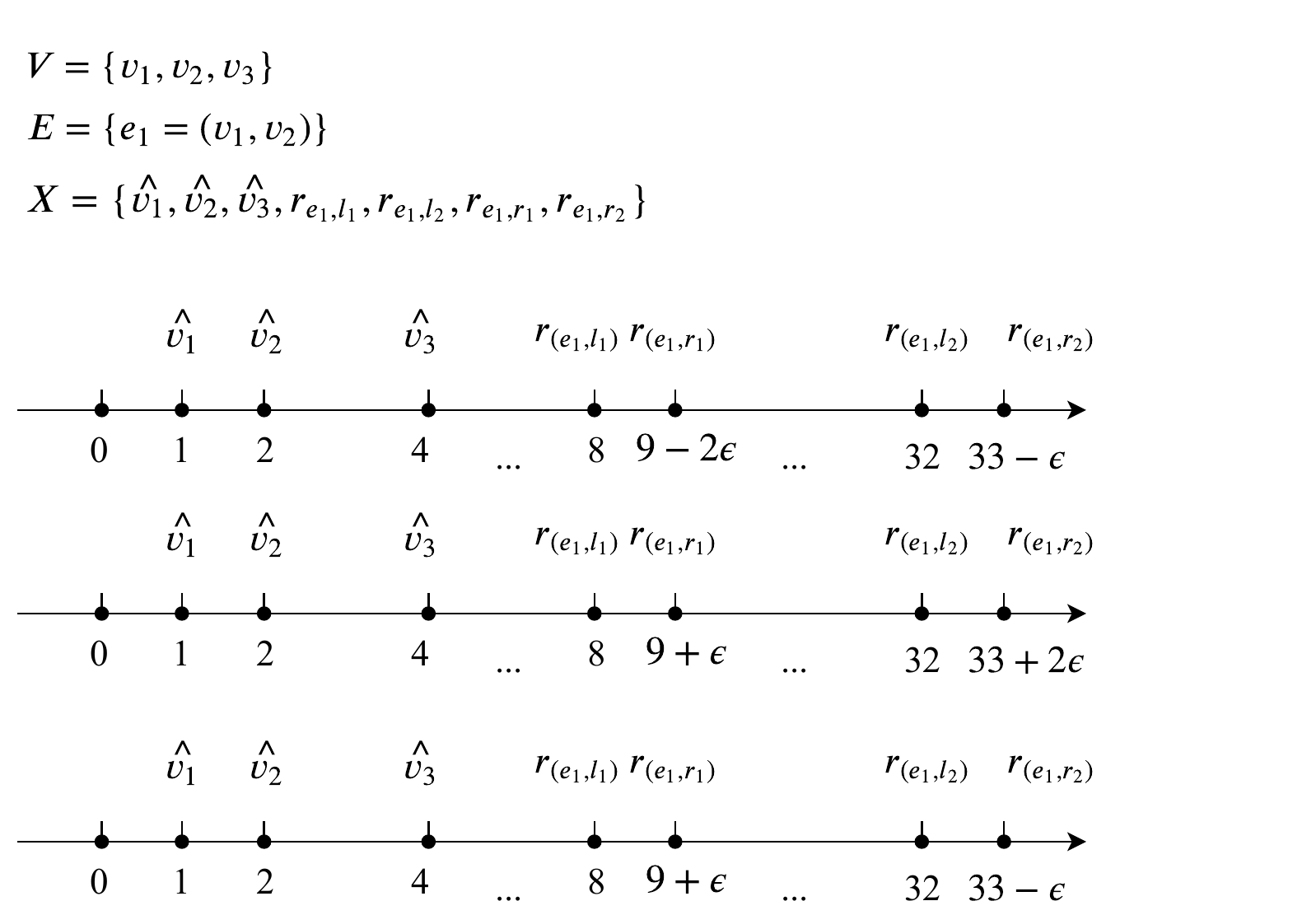}
            \caption{An example of constructed three line metrics given graph $G = (V, E)$. }
            \label{fig:weak_3_line}
        \end{figure}
        
        \noindent \textbf{Comparable pairwise distances.}
        Similarly as Theorem \ref{thm:MSI_2_line_mets}, we will use the concept of \emph{comparable distances}. Here two distances are comparable if they are \emph{different} and differ by at most $6 \epsilon$. The constant factor for $\epsilon$ is changed accordingly since the difference between the coordinates of the same point in two metrics is at most $3 \epsilon$ (instead of $2\epsilon$ in Theorem \ref{thm:MSI_2_line_mets}). The leading terms of $\vhat_i$s, $r_{e, l_1}$ (or $r_{e, r_1}$) and $r_{e, l_2}$ (or $r_{e, r_2}$) are in different scales. For any edge $e_k = (v_i, v_j)$ (here we overuse $k$ for simplicity), the comparable pairs are (recall that $(a, b), (c, d)$ are comparable only when either $[b], [d]$ or $[c], [d]$ have the same leading term):
        \begin{enumerate}
            \item $(\vhat_i, \vhat_j)$, $(r_{e_k, l_1}, r_{e_k, r_1})$ and $(r_{e_k, l_2}, r_{e_k, r_2})$. The distances are around $2^{j-1} - 2^{i-1}$.
            \item $(\vhat_i, r_{e_k, l_1})$ and $(\vhat_j, r_{e_k, r_1})$. The distances are around $2^{n+4k-4} - 2^{i-1}$.
            \item $(\vhat_i, r_{e_k, l_2})$ and $(\vhat_j, r_{e_k, r_2})$. The distances are around $2^{n+4k-2} - 2^{i-1}$.
            \item $(r_{e_k, l_1}, r_{e_k, l_2})$ and $(r_{e_k, r_1}, r_{e_k, r_2})$. The distances are around $2^{n+4k-2} - 2^{n+4k-4}$.
        \end{enumerate}
        
        \noindent \textbf{A triangle can potentially be directed only when all three underlying edges are mutually comparable.}
        
        For any triangle, there are three cases in total depending on the scale of distances of three underlying edges $(x_1, x_2), (y_1, y_2), (z_1, z_2) \in X \times X$.
        \begin{enumerate}
            \item If the distances of these three pairs all are not comparable, the $\epsilon$-distortion will not affect the orders of the distances at all. Therefore, the three metrics should agree on the ordering of these pairwise distances. And clearly, the triangle is not directed. 
            \item If two of them have comparable distance, e.g., $(x_1, x_2)$ and $(y_1, y_2)$. Then w.l.o.g., assume $(z_1, z_2)$ has much larger distance compared with those two. Then the edge in the auxiliary graph from $(z_1, z_2)$ to the other two are all outgoing, thus the triangle is not directed. Otherwise, if $(z_1, z_2)$ has a smaller distance, the edges are all incoming and the triangle is also not directed.
            \item Only when all three edges have comparable distances, there will be possibility that the triangle is directed. If we look at the triangle formed by $\{(\vhat_i, \vhat_j), (r_{e, l_1}, r_{e, r_1}), (r_{e, l_2}, r_{e, r_2}) | e \in E, e = (v_i, v_j)\}$. We have relations $(\vhat_i, \vhat_j) < (r_{e, l_1}, r_{e, r_1})$, $(r_{e, l_1}, r_{e, r_1}) < (r_{e, l_2}, r_{e, r_2})$ and $(r_{e, l_2}, r_{e, r_2}) < (\vhat_i, \vhat_j)$ by plurality vote. And they indeed form a directed triangle from $(\vhat_i, \vhat_j)$ to $(r_{e, l_2}, r_{e, r_2})$ to $(r_{e, l_1}, r_{e, r_1})$, and back to $(\vhat_i, \vhat_j)$.
        \end{enumerate}
        
        Therefore, the directed triangles are $\{(\vhat_i, \vhat_j), (r_{e, l_1}, r_{e, r_1}), (r_{e, l_2}, r_{e, r_2}) \, | \, e = (v_i, v_j) \in E\}$.
        Including $\vhat_i$s in the inconsistent set is better than including $r_{e, l}$ or $r_{e, r}$ since it can possibly cover multiple directed triangles. Thus, there is always an optimal solution consisting of only $\vhat_i$s. One can always replace $r_{e, l}$ and $r_{e, r}$s by $\vhat_i$s (where $\vhat_i \in e$) and still cover the directed triangles. Now we are ready to show that there is a vertex cover of size $K$ for $G$ if and only if there is a weak inconsistent set of size $K$ for $(\mathcal{M}=\{\Metric_1, \Metric_2, \Metric_3\}; X)$.
        
        \noindent \textbf{$``\Rightarrow"$ direction:}
        Assume $V' = \{v_1, \cdots, v_K\}$ is a vertex cover for $G$. Then clearly $\Shat = \{\vhat_1, \cdots, \vhat_K\}$ is a weak inconsistent set. This is because $\Shat$ covers all potential directed triangles.
        
        \noindent \textbf{``$\Leftarrow$" direction:}
        We showed that there is always an optimal solution only consisting of $\vhat_i$s. Assume the optimal solution is $\Shat^* = \{\vhat_1, \cdots, \vhat_K\}$. To cover all directed triangles, for any edge $e = (v_i, v_j) \in E$, there must be at least one corresponding node $\vhat_i$ or $\vhat_j$ contained in the inconsistent set. Therefore, $V' = \{v_1, \cdots, v_K\}$ is a vertex cover of $G$.
        
    \end{proof}

\subsection{Proof of Lemma \ref{lem:nomix}}
\label{appendix:lem:nomix}

\paragraph*{Proving statement (i) of Lemma \ref{lem:nomix}.}    
To prove (i), note that if $w = (p, q)$ splits, then $LCA(p, q)$ is the root $r$ in each of the three representing trees, and thus $U_?(p, q) = 10$, for $?\in \{X, Y, Z\}$. Now take any two other graph nodes $w_1 = (p_1, q_1), w_2 = (p_2, q_2)\in P \times P$. A simple case analysis shows that no matter these two other pairs split or not, the three graph nodes $w, w_1, w_2$ cannot form a directed triangle. For example, if neither $w_1, w_2$, then there will be two directed edges $(w, w_1)$ and $(w, w_2)$ in the auxiliary graph $\mathcal{G}$, and thus no matter what direction the edge between $w_1$ and $w_2$ is, the triangle $w, w_1, w_2$ cannot be directed. This proves (i).

\paragraph*{Proving statement (ii) of Lemma \ref{lem:nomix}.} 
Note that by (i), to check whether auxiliary graph $\mathcal{G}$ 
has directed triangles or not, we only need to consider the subgraph $\widehat{\mathcal{G}}$ of $\mathcal{G}$ spanned by nodes which {\bf do not split}. 
Hence each node $(p, q)$ will be such that $p, q\in \boxtimes$ where $\boxtimes$ 
 could be $\myA, \myB, \myC, \myD$ or $\myR$; in this case, we say that $(p,q)$ is {\it of type-$\boxtimes$}. 
 Two nodes are {\it of the same type} if they are both of type-$\boxtimes$ (i.e, both of type-$\myA$, $\myB$, etc). 
 
 For such a generic non-splitting node $(p,q)$, we list its distance under ultrametric $U_X, U_Y$, and $U_Z$, respectively, which is also the height of the $\LCA(p,q)$ in representing trees $T_X, T_Y$ and $T_Z$, respectively. See Figure \ref{fig:undirected_subgraph_vis}. In particular, for a non-splitting node $(\rhat, \rhat') \in \myR$, we first only consider the case that $\rhat, \rhat'$ {\bf do not} share coordinate. (The only non-splitting nodes not covered by the table in Figure \ref{fig:undirected_subgraph_vis} is of the form $(\rhat, \rhat')$ where $\rhat$ and $\rhat'$ have shared coordinate.) 

\begin{figure}[htb]
    \centering
    \begin{minipage}[c]{0.65\textwidth}
        \centering
        \begin{tabular}{|c|ccc|}
            \hline
            non-splitting graph node &$T_X$ & $T_Y$ & $T_Z$ \\
            \hline\hline
            $(a, a') \in \myA$ & 5 & 3 & 2    \\
            $(b, b') \in \myB$ & 4 & 0 & 4    \\
            $(c, c') \in \myC$ & 3 & 2 & 0    \\
            $(d, d') \in \myD$ & 1 & 5 & 1    \\
            $(\rhat, \rhat')\in \myR$ with {\bf no} shared coordinate & 2 & 4 & 5    \\
            \hline
        \end{tabular}
    \end{minipage}
    \begin{minipage}[c]{0.3\textwidth}
        \includegraphics[width=\textwidth]{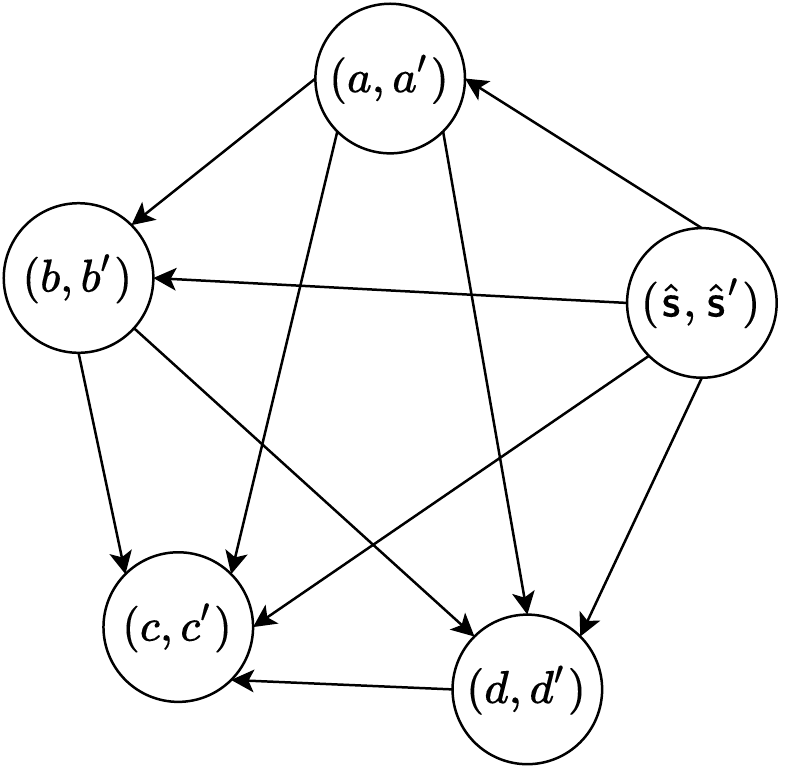}
    \end{minipage}
    \caption{Left: each entry in the table gives the distance (height of \LCA of this pair of points) of a non-splitting graph node under ultrametrics $U_X$, $U_Y$ and $U_Z$, respectively (i.e., w.r.t. trees $T_X$, $T_Y$ and $T_Z$, respectively). Right: the subgraph induced such nodes in the left cannot contain any directed triangle. 
    }
    \label{fig:undirected_subgraph_vis}
\end{figure}

 Based on these distances, the subgraphs induced by such nodes are shown in the right picture of Figure \ref{fig:undirected_subgraph_vis}. 
 Note that this subgraph does not contain any directed triangle. 
 Combining with (i), this means that any directed triangles $\Delta w_1 w_2 w_3$, where $w_i$s are all \emph{of different types}, has to contain a node $(\rhat, \rhat')$ such that $\rhat, \rhat'$ have {\bf shared coordinate(s)}. 
 
 To prove (ii), what remains is to consider a triangle $\Delta w_1 w_2 w_3$ where at least two of them, say $w_1$ and $w_2$ are of the same type. Suppose all $w_i$s are of types listed in the table of Figure \ref{fig:undirected_subgraph_vis}, then it is easy to see that in this case, the resulting triangle can only be of the shapes in Figure \ref{fig:two_same_type_3_triangles}, and thus cannot be a directed triangle. Hence at least one of $w_i$ has to be non-splitting (by statement (i)), yet not included in the types covered in the table of Figure \ref{fig:undirected_subgraph_vis} -- in other words, at least one $w_i$ is of the form $(\rhat, \rhat')$ where $\rhat, \rhat'$ have {\bf shared coordinate(s)}. 
 
  \begin{figure}[thbp]
        \centering 
        \includegraphics[width=0.8\textwidth]{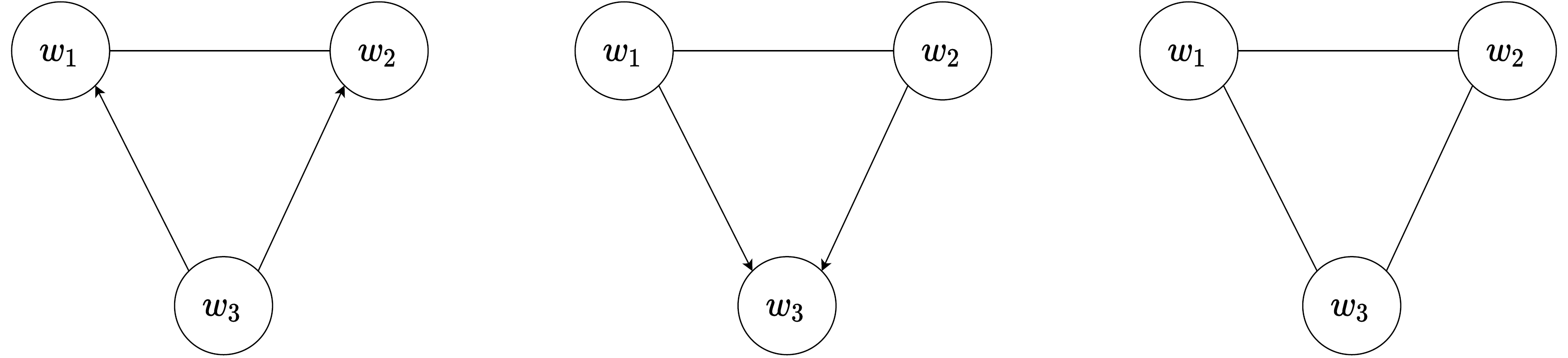}
        \caption{Three possible directed triangles.}
        \label{fig:two_same_type_3_triangles}
    \end{figure}

Putting the above two paragraphs together, statement (ii) then follows. 
 
 \paragraph*{Proving statement (iii) of Lemma \ref{lem:nomix}.} 
 
 Finally, consider a node $w = (\rhat, \rhat')$ of type-$\myR$, but such that $\rhat$ and $\rhat'$ have shared coordinates. 
In particular, $\rhat$ and $\rhat'$ could share $x$-coordinate, share $y$-coordinate, share $z$-coordinate, share both $x$- and $y$-coordinates, share both $x$- and $z$-coordinates, or share both $y$- and $z$-coordinates.  

\begin{figure}[htb]
    \centering
    \begin{minipage}[c]{0.65\textwidth}
        \centering
        \begin{tabular}{|c|ccc|}
            \hline
            non-splitting graph node &$T_X$ & $T_Y$ & $T_Z$ \\
            \hline\hline
            $(a, a') \in \myA$ & 5 & 3 & 2    \\
            $(b, b') \in \myB$ & 4 & 0 & 4    \\
            $(c, c') \in \myC$ & 3 & 2 & 0    \\
            $(d, d') \in \myD$ & 1 & 5 & 1    \\
            \makecell{$(\rhat, \rhat')\in \myR$ with shared \\ $y$ and $z$ coordinates }& 2 & 1 & 3    \\
            \hline
        \end{tabular}
    \end{minipage}
    \begin{minipage}[c]{0.3\textwidth}
        \includegraphics[width=\textwidth]{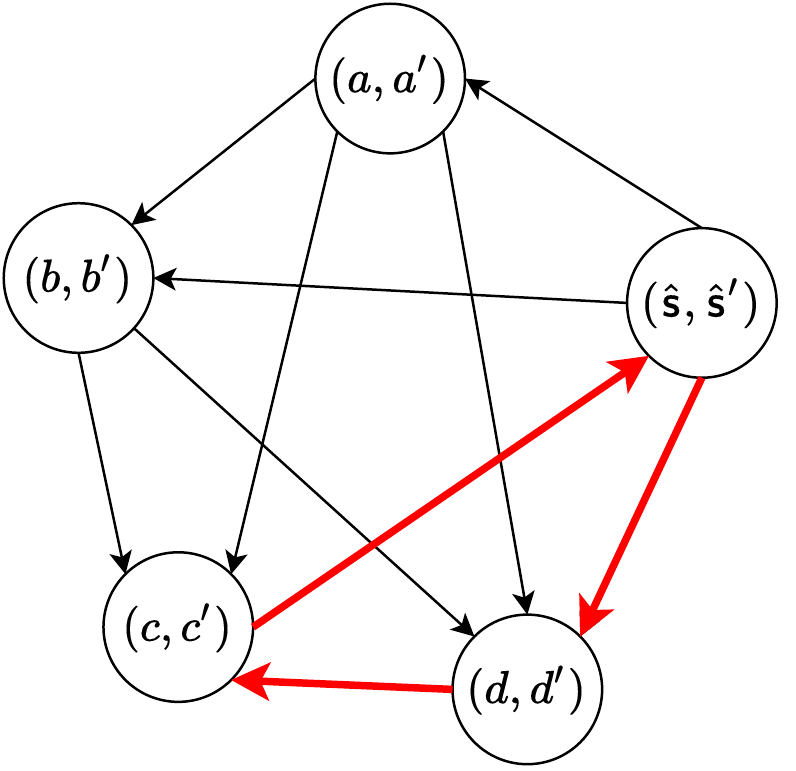}
    \end{minipage}
    \caption{Left: each entry in the table gives the distance (height of \LCA of this pair of points) of a non-splitting graph node under ultrametrics $U_X$, $U_Y$ and $U_Z$, respectively (i.e., w.r.t. trees $T_X$, $T_Y$ and $T_Z$, respectively). Right: the subgraph induced such nodes in the left contains a directed triangle (highlighted in red) formed by $w = (\rhat, \rhat'), (c, c'), (d, d')$. 
    }
    \label{fig:directed_subgraph_shared_yz}
\end{figure}

By a simple but tedious case analysis, one can verify that in each of these 6 cases, $w=(\rhat, \rhat')$ will form a directed with some pair of nodes from $w_1, w_2 \in \{ (a, a'), (b, b'), (c, c'), (d, d')\}$. 
For example, suppose $\rhat$ and $\rhat'$ share both $y$- and $z$- coordinates. 
Then the subgraph induced by $w=(\rhat, \rhat')$ and $\{ (a, a'), (b, b'), (c, c'), \\ (d, d')\}$ is shown in Figure \ref{fig:directed_subgraph_shared_yz}, and the triangle formed by $w$ and $(c, c'), (d, d')$ is a directed triangle. This thus proves statement (iii). 

As commented in the main text, the heights of these nodes are computed by a computer program to guarantee the three statements. 

\subsection{Proof of Theorem \ref{thm:MWI_arbitrary_3_mets}} \label{app:prove_thm_weakNP_3metric}
    
    \begin{proof}
        We prove the theorem via a reduction from \minVertexCover to \WMIS.
        
        \noindent \textbf{Description of the reduction.}
        Given an instance of \minVertexCover, $G = (V, E)$, $V = \{v_1, ..., v_n\}$ and $E = \{e_1, ..., e_m\}$. We construct an instance $(\mathcal{M}=\{\Metric_1, \Metric_2, \Metric_3\}; X)$ of \WMIS with 3 metrics on node set $X = \{r_{e_1}, ..., r_{e_m}, \vhat_1, ..., \vhat_n\}$. We will assign an index number for each edge, and will use that index to assign distances between nodes incident to that edge.
        
        For any edge $e_k = (v_i, v_j)$, we construct the following gadget (Figure \ref{fig:weak_gadget_3mets}), where $M = 3m$ and $\epsilon \ll 1$.
        
        \begin{figure}[ht]
            \centering 
            \includegraphics[width=0.75\textwidth]{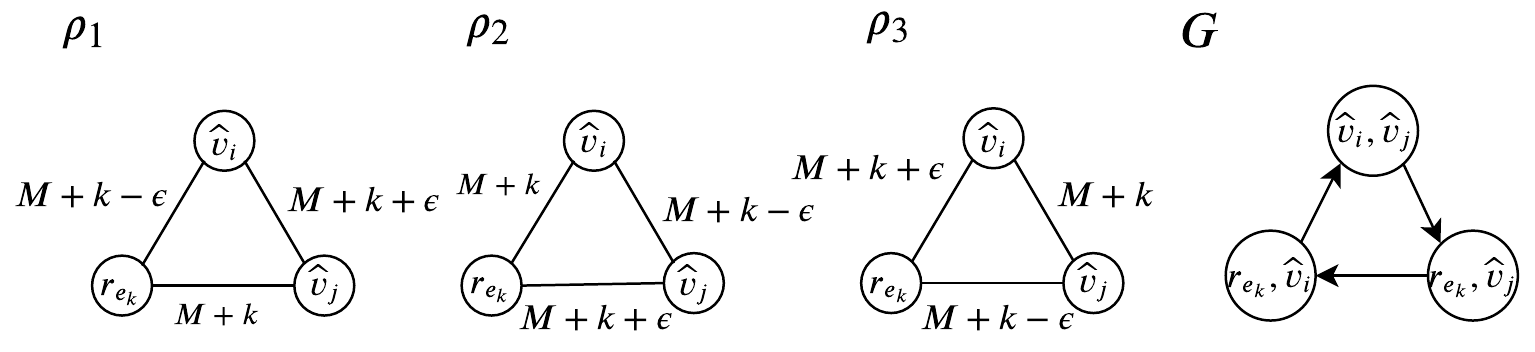}
            \caption{A directed triangle.}
            \label{fig:weak_gadget_3mets}
        \end{figure}
        
        All the other pairwise distances are $2M$. One can check that, in the auxiliary graph $\mathcal{G}$, a directed triangle (if exists) always corresponds to an edge $e = (v_i, v_j)$, and the triangle is consist of $(\vhat_i, \vhat_j), (\vhat_i, r_e)$ and $(\vhat_j, r_e)$. It is because that the only comparable distances are triples of form $\{(\vhat_i, \vhat_j), (\vhat_i, r_e), (\vhat_j, r_e) \}$. Now we prove that there is an vertex cover of size $K$ for $G$ if and only if there is a (weakly) inconsistent set of size $K$ for $(\mathcal{M}=\{\Metric_1, \Metric_2, \Metric_3\}; X)$. 
        
        \noindent \textbf{$``\Rightarrow"$ direction:}
         Assume for graph $G$, there is a vertex cover $S = \{v_1, ..., v_K\}$. Then $\widehat{S} = \{\vhat_1, ..., \vhat_K\}$ is a \WMIS, since each directed triangle corresponds to an edge in $G$, and has at least one relevant node in $S$ ($\widehat{S}$).
        
        \noindent \textbf{$``\Leftarrow"$ direction:}
        If we have a \WMIS of size $K$, denoted as $\widehat{S}^* = \{\vhat_1, ... \vhat_K\}$ (there is always an optimal solution only consisting of $\vhat_i$s). There is no directed triangle means that for each edge, at least one endpoint is in $S =\{v_1, ..., v_K\}$, which is a vertex cover of size $K$.
        
    \end{proof}

 \subsection{Summary of Hardness Results}  \label{app:hardness_results}  
        We end this section with a summary table showing hardness results in different cases.
        
        \begin{table}[ht] 
            \centering
             \resizebox{\textwidth}{!}{
            \begin{tabular}{|ccccc|}
            \hline
                Problem & Input  & Hardness & \makecell{Inapprox. \\ factor} & \makecell{Approx.\\ algorithm} \\
            \hline
                \SMIS & 2 Arbitrary metrics & \makecell{NP-complete \\ (Theorem \ref{thm:strong_NP_2metrics})} & \makecell{2 UGC \\ (Corollary \ref{coro:strong_NP_2metrics_inapprox})} & \makecell {4\\ (Theorem \ref{thm:MSI_4_approx_alg})} \\
                \SMIS & 2 Line metrics & \makecell{Weakly NP-complete\\ (Theorem \ref{thm:MSI_2_line_mets})} & ? & 4 \\
                \SMIS & 2 Ultrametrics & \makecell{NP-complete\\ (Theorem \ref{thm:MSI_2_ultra_NP})} & ? & 4 \\
                \WMIS & 3 Arbitrary metrics & NP-complete & \makecell{2 UGC \\ (Theorem \ref{thm:MWI_arbitrary_3_mets})} & \makecell{6 \\ (Theorem \ref{thm:MWI_6_approx})}\\
                \WMIS & 3 Line metrics & \makecell{Weakly NP-complete \\  (Theorem \ref{thm:MWI_3_line_metrics})} & ? & 6 \\
                \WMIS & 3 Ultrametrics & \makecell{NP-complete \\ (Theorem \ref{thm:MWI_3_ultra_NP})} & ? & 6 \\
                \hline
            \end{tabular}
            }
            \caption{Hardness results for different cases.} \label{tab:hardness_summary}
        \end{table}

\section{NP-complete Problem Repository} \label{app:np_complete_repo}
Here we list the NP-complete problems that are used in the hardness proofs.
    
    \begin{definition}[\minVertexCover] \cite{NP_problems_Karp1072} \\
        \textbf{Instance:} Graph $G = (V, E)$ and a positive integer $K$.\\
        \textbf{Question:} Does $G$ have a vertex cover of size at most $K$?
    \end{definition}
    A vertex cover is a set of nodes $V' \subset V$ that every edge has at least one endpoint in $V'$. This is a classical problem mentioned in Karp's 21 np-complete problems.

    \begin{definition}[\maxTwoSAT] \cite{Max_2_Sat_Garey_1974}\\
        \textbf{Instance:} Given a boolean expression $E$ of $n$ variables in conjunctive normal form (CNF) that is the
        conjunction of $m$ clauses over $n$ variables, each of which is the disjunction of at most two distinct literals.  An integer $K$.\\
        \textbf{Question:} Is there an assignment to variables such that $K$ clauses are satisfied? 
    \end{definition}

    \begin{definition}[\tournamentFeedbackVertexSet] \cite{3ordinal_directed_tri_Dom2010}\\
        \textbf{Instance:} A tournament (fully connected directed graph) $T$ and an integer $K$.\\
        \textbf{Question:} Is there a vertex set $S$ with at most $K$ nodes whose deletion will result in an acyclic directed graph.
    \end{definition}
    In \cite{3ordinal_directed_tri_Dom2010}, Dom showed that Feedback Vertex Set problem is NP-complete.

    \begin{definition}[\threeDimensionalMatching] \cite{NP_problems_Karp1072}\\
        \textbf{Instance:} Let $X, Y, Z$ be three disjoint sets with the same size. And $\matchS \subset X \times Y \times Z$ consists of triples $(x, y, z)$ where $x \in X, y \in Y, z \in Z$. Given $(X, Y, Z; \matchS)$ and an integer $K$. \\
        \textbf{Question:} Is there a 3-dimensional matching $\Pi \subset \matchS$ with size at least $K$?
    \end{definition}
    This is also a problem in the list of Karp's 21 np-complete problems.

\end{document}